\documentclass[12pt]{article}
\usepackage{amsmath, amssymb, amsthm}
\usepackage{geometry}
\usepackage[hidelinks]{hyperref}
\usepackage{enumitem}
\usepackage{parskip}
\usepackage{booktabs}
\usepackage{longtable}
\usepackage[table]{xcolor}

\newtheoremstyle{spaced}
  {12pt}       
  {12pt}       
  {\itshape}   
  {}           
  {\bfseries}  
  {.}          
  {.5em}       
  {}           
\theoremstyle{spaced}
\newtheorem{theorem}{Theorem}
\newtheorem{definition}{Definition}
\newtheorem{proposition}{Proposition}
\newtheorem{corollary}{Corollary}
\newtheorem{lemma}{Lemma}
\newtheorem{assumption}{Assumption}
\title{Institutional Harm Through\\ Threshold Cascades}
\author{Piper Harris \and Chad M. Topaz}
\date{\today}
\begin{document}
\maketitle
\begin{abstract}
Can a population of people not individually inclined to harm others nonetheless produce harmful collective outcomes, purely because of the institutional structure they inhabit?  Social scientists have long argued yes, but existing accounts are largely qualitative and provide no precise condition distinguishing safe institutions from unsafe ones.  We develop a threshold cascade model in which agents have positive activation thresholds, harmful behavior is irreversible, and the institution exerts both standing pressure and peer influence along a weighted network.  We give a necessary and sufficient condition, checkable from the institution's structure and its members' thresholds, for resistance to any shock up to a given size; the criterion extends to signed influence, in which some peer effects counteract harm, and yields a convex optimization formulation for least-cost repair.  It also reveals a sharp frontier between functionality and safety.  An institution can coordinate its members and remain safe if and only if the exposure that coordination creates stays below the weakest member's net threshold.  A further tension arises when coordination requires responsiveness to peer influence, which can make it impossible to prevent the most exposed group from cascading.  We then analyze a mean-field model of two groups differing in how easily their members are pushed into harm.  When one group is unstable in isolation but the system is stable under full mixing, disproportionate within-group influence creates a sharp homophily threshold beyond which the harm-free state becomes unstable.  In the model, identical treatment of both groups does not generally equalize their cascade robustness.
\end{abstract}
\newpage
\section{Introduction}
\label{sec:intro}

People are regularly harmed by institutions---by courts, police departments, banks, universities, and corporations, among others.
Sometimes the harm reflects the institution functioning exactly as designed; sometimes it reflects deviation from stated goals.  In either case, a large body of social science research suggests that harm often arises even when many of the individuals inside these institutions are not personally inclined toward it \cite{milgram1963, browning1992}.
If that premise is accepted, a critical question follows.  How does institutional harm emerge from a population that is not individually disposed to produce it?

One common answer is that the problem lies in bad actors within the institution.  If so, the remedy would seem simple.  Identify the bad actors and remove them.  The appeal is immediate, but the explanation does not hold up well.
The ``bad actor'' account says nothing about the structure of institutions themselves.  It does not explain why some institutions consistently produce harm and others do not, why harm falls along lines of race or class rather than striking randomly, or why replacing individual people so rarely fixes the problem.  Often institutions identify individual wrongdoers but decline to remove them, a refusal that is itself part of the structural pattern the bad-actor account does not address.

A substantial literature in social science argues that institutional structure, rather than individual malice, can itself generate systematic harm.  This idea is central to work on structural violence \cite{galtung1969}, organizational wrongdoing \cite{ashforth2003}, and normalization of deviance \cite{vaughan1990, vaughan1996}.  It also underpins research on street-level bureaucracy, in which frontline workers such as police officers, teachers, and welfare caseworkers exercise discretionary power over the people they serve \cite{lipsky1971, lipsky1980}.  These accounts make a powerful qualitative case.  Harmful outcomes can emerge from authority structures, role expectations, incentive systems, and organizational norms even when many of the individuals inside the institution are not personally committed to harm.

But qualitative diagnosis is not the same as mathematical control.  These accounts explain \emph{why} structure matters, but they do not by themselves tell us \emph{when} an external disturbance that pushes some members into harmful behavior will remain localized and \emph{when} it will spread through the institution.  Nor do they tell us \emph{what} structural changes would make an institution robust to such disturbances.  We aim to supply that missing formal layer.  The question we formalize is:

\begin{quote}
\emph{If the individual members of an institution are not inclined to harm, what structural properties of the institution can produce harmful collective outcomes?  And conversely, what structural properties prevent harm from spreading when a small group of agents is pushed into harmful behavior by an external shock?}
\end{quote}

We make three contributions.

The first is an exact theory for a finite population of agents.  We model an institution as a network in which each person has a personal threshold for engaging in harmful behavior, and harmful behavior, once adopted, is not reversed.  An external shock activates a small group of agents, and the question is whether the resulting peer pressure causes harm to spread through the network.  We give a precise, checkable condition on the institution's structure and its members' thresholds that determines whether it can withstand any shock up to a given size.

The second contribution uses this condition to study institutional repair and a tradeoff between functionality and safety.  Repair asks what is the least costly combination of reducing baseline pressure and weakening peer influence that restores safety.  We show this repair problem has a convex structure, meaning it can be solved efficiently without getting stuck at false solutions.  The tradeoff result says that requiring agents to be sufficiently responsive to one another---a basic condition for the institution to function---can be mathematically incompatible with preventing the institution's most exposed members from being cascaded into harmful behavior. 

The third contribution is a complementary large-population model.  We replace individual agents with continuous distributions in a mean-field approximation, following the standard modeling practice of analyzing the same phenomenon at both the agent level and the population level.  This approach lets us study two groups within the same institution and ask how differences in susceptibility to engaging in harm interact with the tendency of people to be more influenced by members of their own group.

The main mathematical novelty is a criterion that determines, from the institution's structure and its members' thresholds, whether any group of up to $k$ externally activated agents can trigger a cascade.  Because harmful behavior is not reversed once adopted, first-step prevention is sufficient.  If no non-seed agent activates at $t = 1$, no cascade ever starts.  The criterion uses the $k$ strongest incoming influences on each agent and gives an exact bound on worst-case first-step exposure.  It extends to settings where some peer influences are protective rather than harmful.

\section{Background}
\label{sec:background}

A \emph{threshold model} is a mathematical model in which each agent in a population has a personal threshold for adopting some behavior.  The agent adopts the behavior if and only if enough of their peers have already done so.  This idea was developed to explain collective social phenomena such as riots, bank runs, technology adoption, and norm change \cite{granovetter1978, schelling1971}.  Suppose there are $n$ agents, and each agent $i$ has a threshold $\tau_i \geq 0$.  Agent $i$ adopts the behavior when the number of adopters reaches or exceeds $\tau_i$.
In Granovetter's formulation, each peer who has adopted contributes equally---one unit---toward meeting $\tau_i$.  Our model (Section~\ref{sec:model}) generalizes this by allowing each peer $j$ to contribute a different amount of pressure $w_{ij}$, so that $\tau_i$ becomes a threshold on total weighted pressure rather than a simple head count.  Starting from some initial set of adopters, we ask: does adoption spread to the whole population, fizzle out, or stabilize at some intermediate level?  Granovetter's key insight was that the answer depends on the \emph{distribution} of thresholds, not just the average.  A population with a high average threshold can still cascade into full adoption if the distribution has the right shape, with enough agents at low thresholds to start the process and enough at moderate thresholds to sustain it.

A concrete example makes this precise.  Consider ten agents whose thresholds are
\begin{equation}
(\tau_1, \tau_2, \ldots, \tau_{10}) \;=\; (0,\; 1,\; 2,\; 3,\; 4,\; 5,\; 6,\; 7,\; 8,\; 9).
\end{equation}
Each agent watches the entire population and adopts the behavior as soon as the total number of current adopters meets or exceeds their personal threshold.  The average threshold is $4.5$.  At time $t = 0$, no one has adopted yet, but the agent with $\tau_1 = 0$ requires zero adopters, so that agent adopts immediately.  Now there is one adopter.  At $t = 1$, the agent with $\tau_2 = 1$ sees one adopter, which meets their threshold, so they adopt too.  Now there are two adopters.  At $t = 2$, the agent with $\tau_3 = 2$ sees two adopters and activates.  The process continues in the same way: at each step, the rising count of adopters crosses exactly one more agent's threshold, and the cascade rolls through the entire population in nine steps.  Despite an average threshold of $4.5$, the evenly spaced distribution creates an unbroken chain of activations.

Now consider a different distribution with the same average.  Ten agents have thresholds
\begin{equation}
(\tau_1, \tau_2, \ldots, \tau_{10}) \;=\; (0,\; 0,\; 0,\; 0,\; 0,\; 9,\; 9,\; 9,\; 9,\; 9),
\end{equation}
again averaging $4.5$.  At $t = 0$, the five agents with $\tau_i = 0$ all adopt immediately, producing five adopters.  But each of the remaining five agents has $\tau_i = 9$ and requires nine adopters to activate.  With only five adopters in a ten-person population, none of these thresholds can be met, at this step or any future step.  The cascade stalls permanently at five out of ten---half the population---compared to the complete cascade in the first example.  The average threshold is the same in both cases, but the shape of the distribution determines the outcome.  The evenly spaced distribution creates a ladder that the cascade can climb one rung at a time; the bimodal distribution creates a gap that the cascade cannot cross.

The dynamics of a threshold model have a natural analogy in epidemiology.  In a standard epidemic model, the key quantity is the \emph{basic reproduction number} $R_0$: the expected number of secondary infections produced by a single infected individual in an otherwise susceptible population.  If $R_0 > 1$, the infection spreads exponentially, at least initially.  If $R_0 < 1$, it dies out.  We will derive an analogous threshold condition for institutional harm, a quantity $\lambda_0$ that plays the role of $R_0$ in our mean-field approximation.

The mathematical study of threshold-driven cascades now spans several distinct research streams, and our framework builds on, and differs from, each.  The five streams most relevant to our work are well-mixed threshold models, network cascade models, simple-versus-complex contagion, informational cascades and strategic adoption, and influence maximization.
The first stream, originating with Granovetter~\cite{granovetter1978} and Schelling~\cite{schelling1971}, asks how the distribution of individual thresholds determines collective outcomes.  In this line of work, agents sit in a well-mixed population, each watches the aggregate adoption level, and the question is whether a small seed of adopters triggers a chain reaction.  Our finite-agent model inherits this basic logic but replaces the well-mixed assumption with an explicit influence matrix $W$ that encodes who affects whom and by how much.

The second stream places threshold dynamics on networks, where who is connected to whom matters.  On random graphs, there is a \emph{cascade window}: a range of network connectivity within which a single adopter can trigger a population-wide cascade.  Below this window, the network is too sparse to transmit influence; above it, the network is so dense that no single adoption can push any neighbor over the edge \cite{watts2002}.  A related game-theoretic analysis considers settings in which each agent's best course of action depends on what their neighbors are doing, so that adopting a behavior is rational only when enough neighbors have already done so.  Whether contagion can then spread from a small initial group to the entire network depends on a threshold quantity determined by the pattern of local connections \cite{morris2000}.  Both of these results apply to networks in which every connection has the same strength.  The model of \cite{watts2002} allows heterogeneous fractional thresholds but assumes uniform edge weights, and the analysis of \cite{morris2000} uses identical thresholds.  Our Theorem~\ref{thm:topk} extends this line of work to networks with heterogeneous weights and heterogeneous absolute thresholds simultaneously, and provides a guarantee that applies to any group of up to $k$ initial adopters, not just a single one.

A third stream concerns the distinction between simple and complex contagion.  Many social phenomena require reinforcement from multiple independent sources before an agent adopts, unlike biological diseases, where a single contact may suffice \cite{centola2007}.  This distinction matters for network structure: long-range ties that accelerate simple contagion may be ineffective or even harmful for complex contagion, because they deliver single exposures that fall below the reinforcement threshold.  A generalized contagion model smoothly bridges these two extremes by letting agents accumulate exposure over a window of recent contacts, with different agents requiring different amounts of cumulative exposure before they adopt \cite{dodds2005}.  At one extreme the model reproduces disease-like contagion in which a single contact can infect; at the other it reproduces threshold-driven contagion requiring many reinforcing contacts.  Between these extremes, the model identifies three qualitatively distinct classes of spreading behavior.  Our model sits at the threshold end of this spectrum: an agent activates only when the total influence pressure from active peers exceeds a personal threshold, and once activated, the agent does not revert.

A fourth stream addresses the strategic and economic side of cascades.  In models of \emph{informational cascades} and \emph{herd behavior}, agents observe predecessors acting a certain way, conclude those predecessors must know something, and follow suit regardless of their own private signals \cite{bikhchandani1992, banerjee1992}.  This mechanism is distinct from threshold-based adoption but produces similar collective dynamics.  The key difference is that informational cascades are fragile.  A small amount of public information can shatter them.  In our model, by contrast, cascades are self-reinforcing and irreversible.  Once an agent activates, they remain active, and the question is what structural conditions prevent the cascade from starting rather than what breaks it after it has begun.

A separate line of work in this fourth stream studies how behavior spreads when each agent's incentive to adopt depends on how many of their connections have already done so \cite{jackson2007}.  The distribution of connection counts across agents plays a decisive role.  Networks mixing highly connected and sparsely connected agents behave very differently from uniform networks, and this heterogeneity determines when a small wave of adoption tips into a large one.  This analysis connects thematically to our mean-field extension in Section~\ref{sec:mean-field-fairness}, though our mean-field model assumes uniform mixing.  In our framework, heterogeneity comes from the threshold distribution, not from degree variation, and we focus on the conditions under which the harm-free state first becomes unstable.

A fifth stream concerns influence maximization.  If you can choose a small number of agents to activate first, which ones should you pick to make the behavior spread as widely as possible?  This selection problem is computationally intractable in general---there is no known efficient algorithm that finds the optimal set.  Nevertheless, a natural greedy strategy that repeatedly picks the single agent who adds the most new activations comes with a provable performance guarantee \cite{kempe2003}.  Our work inverts this question.  Rather than choosing seeds to maximize spread, we ask what structural conditions prevent \emph{any} set of up to $k$ seeds from spreading.  Theorem~\ref{thm:topk} gives a checkable answer, and the repair theory of Section~\ref{sec:repair} asks how to modify institutional structure so that the answer is ``no cascade.''

Our contributions differ from existing results in three respects.  Existing results apply to networks with uniform or binary weights \cite{watts2002, morris2000}, to well-mixed populations without network structure \cite{granovetter1978, schelling1971}, or to seed-selection optimization on a fixed network \cite{kempe2003}.  None provides a checkable necessary-and-sufficient condition for whether an arbitrary weighted, heterogeneous network can withstand any shock of bounded size.  None asks how to \emph{redesign} the network to achieve robustness, or what structural tradeoff arises when agents must remain responsive enough to coordinate.  Our framework addresses all three gaps with an exact multi-seed robustness criterion (Theorem~\ref{thm:topk} and its signed extension, Theorem~\ref{thm:signed}), a convex repair theory (Theorem~\ref{thm:mixed-repair}), and a sharp coordination--robustness frontier (Theorem~\ref{thm:functionality}).  The mean-field two-group analysis (Section~\ref{sec:mean-field-fairness}) illustrates how the same ideas extend to large populations with group structure and homophily.

\section{The Model}
\label{sec:model}

Our model has four components.  First, each agent has a personal activation threshold that must be exceeded before they engage in harmful behavior.
Second, the institution exerts two kinds of pressure on each agent: a baseline pressure from rules, norms, and hierarchy, and peer influence transmitted through a weighted network.  Third, we impose a separation condition requiring that baseline pressure alone is never enough to push anyone into harm;
cascades can begin only when an external shock combines with peer influence.  Fourth, dynamics are absorbing, meaning that once an agent begins engaging in harm, they do not stop.%

We do not define ``harm'' narrowly.  The model applies whenever an institution can push its members into conduct that injures others, whether through discriminatory lending, excessive use of force, fraudulent reporting, or other forms of harmful conduct documented in Section~\ref{sec:intro}.  What matters formally is the binary distinction between engaging in harmful behavior and not doing so.

All notation used throughout the paper, in both the exact finite-agent analysis and the mean-field extension, is collected for reference in Appendix~\ref{app:notation}.

\subsection{Agents and the prosocial default}
\label{sec:agents}

Consider a population of $n$ agents.  We label them with the index set $V = \{1, 2, \ldots, n\}$, and we write ``agent $i$'' as shorthand for ``the agent labeled $i$.''
The agents represent the members of an institution: employees, officers, administrators, or other role-holders.

Each agent $i$
has an \emph{activation threshold} $\tau_i > 0$, defined as the minimum total pressure---from all sources, both institutional and social---that must be applied to agent $i$ before they engage in harmful behavior.  We decompose this total into baseline and social components in Section~\ref{sec:separation}.  Thresholds vary across agents, reflecting the fact that different people require different amounts of pressure to be pushed toward harm.  This variation is important because cascades may be initiated by low-threshold agents before spreading to higher-threshold ones.

The condition $\tau_i > 0$ for all agents is the \emph{prosocial assumption}:
\begin{assumption}[Prosocial Default]
\label{def:prosocial}
A population of agents satisfies the \emph{prosocial default} if each agent $i$ has an activation threshold $\tau_i > 0$: agent $i$ will not take a harmful action unless the total institutional pressure they experience meets or exceeds $\tau_i$.
\end{assumption}
The prosocial default says that harmful behavior requires a positive amount of institutional pressure to produce.  The threshold $\tau_i > 0$ means only that harm is not the default---some external push is required to get there.  This is compatible with purely self-interested agents, as long as harming others is not their default preference absent institutional pressure.  We treat the prosocial default as a domain-specific behavioral assumption, appropriate for institutional settings in which agents occupy formal roles.  Harm in these settings is an outcome of role behavior under institutional pressure rather than a private goal, and the relevant question is what a given person would do if the pressure were removed.

The prosocial default motivates the model behaviorally, but the theorems below depend on the stronger \emph{separation condition} $\theta_i = \tau_i - \sigma_i > 0$, where $\sigma_i$ is the baseline institutional pressure on agent $i$.  The separation condition says that baseline pressure alone is not enough to push anyone into harm.  We define $\sigma_i$ and the separation condition precisely in Sections~\ref{sec:structure} and~\ref{sec:separation} below, after introducing the institutional structure.

The prosocial default is a modeling assumption, not an empirical claim that every member of a real population has $\tau_i > 0$.  What matters for the model is whether the assumption is reasonable for the institutional setting under study.  Two lines of evidence bear on this question, one from developmental psychology and one from social psychology.  The developmental evidence speaks to the baseline: absent external pressure, is prosocial behavior the default?  The social-psychological evidence speaks to the mechanism: can institutional pressure override that default?

Experimental work in developmental psychology shows that spontaneous helping emerges early, before extensive formal socialization.  When 18-month-old infants watched an adult struggle to achieve a goal, they spontaneously helped in a majority of cases, without being asked, trained, or rewarded \cite{warneken2006}.  The tasks ranged from retrieving a dropped object to opening a cabinet to pulling something through a hole.  Twenty-two of twenty-four infants helped in at least one task, typically within seconds.  Infants helped significantly more in experimental conditions where the adult genuinely needed assistance than in matched control conditions where the adult acted intentionally, ruling out simple imitation.  A follow-up study found that extrinsic rewards \emph{undermined} this spontaneous helping: 20-month-olds who received material rewards for helping subsequently helped \emph{less} than those who received no reward or only verbal praise \cite{warneken2008}.  This overjustification effect---in the youngest sample in which it has been demonstrated---suggests that the earliest helping behaviors are intrinsically motivated and that external incentive structures can erode them.

Further experiments found that even 18-month-olds helped others not only with practical tasks but also in response to emotional distress and at personal cost, giving up their own possessions to comfort an adult \cite{svetlova2010}.  These studies support a narrower claim: spontaneous helping emerges early in development and is sensitive to external incentives.  They motivate the modeling assumption without proving it.

A separate body of evidence from social psychology suggests that institutional structure can push people whose baseline is prosocial into harmful behavior.  Milgram's obedience experiments \cite{milgram1963} found that many subjects administered what they believed were painful electric shocks to strangers when instructed by an authority figure.  Both the ethics of the studies \cite{baumrind1964} and their interpretation \cite{perry2013} have been contested, but the core behavioral finding---that situational authority pressure can elicit harmful compliance from ordinary people---has been replicated in modified designs \cite{burger2009}.  In the language of our model, the finding is that $\tau_i > 0$ (subjects did not harm spontaneously) but that sufficient institutional pressure crossed the threshold.

Browning's study of a German reserve police battalion involved in mass killings during World War II \cite{browning1992} makes a similar case under extreme conditions.  The battalion members were mostly middle-aged, working-class German men, and the study identifies authority, conformity, peer pressure, and institutional context as the activating mechanisms; diffusion of responsibility, in the sense made precise by Darley and Latan\'{e} \cite{darley1968}, is part of the picture.  The Stanford Prison Experiment \cite{haney1973} is sometimes invoked here as well, but its scientific validity is contested \cite{griggs2014, letexier2019}, and we do not rely on it as evidence.

None of this evidence proves that any given real population satisfies the prosocial default.  Developmental studies suggest spontaneous helping is an early-emerging behavior; social-psychological studies suggest institutional pressure can override it.  Together they make the prosocial default a defensible modeling assumption for institutional settings in which agents occupy formal roles and harmful behavior arises from role performance under pressure rather than from private disposition.

The social-psychological evidence involves subjects who recognized the harm they were causing, so it speaks most directly to settings of conscious compliance under pressure.  The activation threshold $\tau_i$ measures resistance to engaging in harmful behavior whether or not the agent recognizes it as harmful, so the framework also applies to normalized harm where agents cross the threshold without awareness.

\subsection{Institutional structure}
\label{sec:structure}

An institution exerts two distinct types of pressure on its members.  \emph{Baseline pressure} comes from the rules, incentives, hierarchy, and culture of the institution itself, independent of what any specific colleague is currently doing.  \emph{Social influence pressure} is the additional pressure an agent experiences when they observe colleagues engaging in harmful behavior.  We formalize these as follows.

\begin{definition}[Institutional Structure]
\label{def:structure}
An \emph{institutional structure} (or simply \emph{structure}) is a pair $\mathcal{S} = (\sigma, W)$, where:
\begin{enumerate}[label=(\roman*)]
    \item $\sigma = (\sigma_1, \ldots, \sigma_n) \in \mathbb{R}_{\geq 0}^n$ is the \emph{baseline pressure vector}.  The component $\sigma_i \geq 0$ is the institutional pressure on agent $i$ from rules, incentives, and hierarchy, independent of the current behavior of other agents.
    \item $W \in \mathbb{R}_{\geq 0}^{n \times n}$ is the \emph{influence matrix}, with $w_{ii} = 0$ for all $i$.  The entry $w_{ij} \geq 0$ is the increase in pressure on agent $i$ when agent $j$ is engaging in harmful behavior; when $j$ is not, agent $j$ contributes zero pressure to $i$.
\end{enumerate}
\end{definition}

The entry $w_{ij}$ measures how much agent $j$'s behavior affects agent $i$.  Note the index order: the \emph{row} index $i$ identifies the agent who \emph{receives} pressure, and the \emph{column} index $j$ identifies the agent who \emph{generates} it.  Column $j$ of $W$ encodes agent $j$'s ``influence profile'': how much pressure $j$'s harmful behavior creates for every other agent.  We set the diagonal entries $w_{ii} = 0$ for all $i$, since an agent does not exert social influence pressure on themselves.

In a hierarchical institution, a supervisor's column in $W$ has many large entries.  In a flat institution, no single column dominates.  The column structure of $W$ encodes the institution's authority and visibility structure.

We initially assume $w_{ij} \geq 0$ for all $i, j$,
so that influence only pushes agents toward harm.  Real institutions also contain protective pressures such as oversight, whistleblowing norms, and peer restraint.  We extend the model to signed influence in Section~\ref{sec:signed}, where the core exact results survive.  We present the nonnegative case first for clarity.

All model quantities---thresholds, pressures, influence weights---are defined relative to a fixed institution and a specific type of harmful behavior.  The threshold $\tau_i$ is agent $i$'s resistance to \emph{that particular behavior in that particular institutional context}, not a universal measure of moral resilience.  A person might have a high threshold for one form of institutional wrongdoing and a low threshold for another.  The prosocial default ($\tau_i > 0$ for all $i$) is therefore a property of the model instance, not a blanket claim about the people in it.

\subsection{Net thresholds and the separation condition}
\label{sec:separation}

The activation threshold $\tau_i$ measures the total pressure needed to activate agent $i$, while the baseline pressure $\sigma_i$ is the portion of that pressure supplied by the institution itself, independent of peer behavior.  The difference $\theta_i = \tau_i - \sigma_i$, which we call the \emph{net threshold} of agent $i$, is the amount of additional social pressure---beyond what the institution already provides---that is needed to push agent $i$ into harmful behavior.
If $\theta_i > 0$, baseline pressure alone is insufficient and social contagion from active peers is also required.  If $\theta_i \leq 0$, the institution's baseline already exceeds the agent's threshold, meaning the agent would engage in harmful behavior even without any peer influence.

We say that a structure $\mathcal{S} = (\sigma, W)$ satisfies the \emph{separation condition} if $\theta_i > 0$ for every agent $i \in V$, so that no one is pushed into harm by baseline pressure alone.  We assume the separation condition throughout the paper unless otherwise noted.  Under separation, the state in which no agent engages in harmful behavior is a valid resting point of the system.  The central question is then: under what conditions can a small perturbation cause that resting point to unravel?

A clarification about what $\theta_i$ measures.  The net threshold is resistance to a specific harmful behavior under specific institutional pressure---not ``social vulnerability'' in the policy sense.  An agent with low $\theta_i$ has low resistance to being pushed into the modeled behavior by the modeled pressure; they may have substantial social power in other respects.  Milgram's subjects are the canonical example: ordinary middle-class adults, not socially vulnerable in any policy sense, who turned out to have low thresholds for following authority pressure into administering shocks.  Throughout the paper, ``low threshold'' should be read in this narrow technical sense, not as a claim about general vulnerability or social status.

\subsection{Dynamics}
\label{sec:dynamics}
Let $x_i(t) \in \{0,1\}$ record the state of agent $i$ at discrete time step $t = 0, 1, 2, \ldots$, and write $x(t) = (x_1(t), \ldots, x_n(t))$ for the full state vector.
The value $1$ means agent $i$ is engaging in harmful behavior at step $t$, and $0$ means they are not.

\begin{definition}[Seed Set]
\label{def:seed}
A \emph{seed set} $S_0 \subseteq V$ is a set of agents activated at time $t = 0$ by an exogenous shock.  Seed activation is \emph{absorbing}: $x_i(t) = 1$ for all $t \geq 0$ and all $i \in S_0$.
\end{definition}

The exogenous shock represents any event outside the model that pushes a small group of agents into harmful behavior before peer influence has a chance to operate.  In institutional settings, such shocks take concrete forms.  A policy directive from above may force a few agents into harmful compliance.  A leadership change may shift norms for those in the new leader's direct orbit.  An external crisis---budget cuts, a political mandate, a market shock---may make harmful shortcuts necessary for a subset of agents.  Or new members may import harmful practices from a previous institution.  The model does not need to explain what caused the initial activation, just as an epidemiological model does not need to explain the origin of the first infection.  Its purpose is to determine whether the initial activation spreads.

The dynamics after the shock are governed by two rules.  First, an inactive agent activates whenever the total pressure on them---baseline institutional pressure plus social influence from currently active peers---meets or exceeds their activation threshold.  Second, once an agent activates, we treat that activation as absorbing on the timescale of the cascade episode under study.  This is not a claim that harmful behavior is literally permanent in all institutional settings.  Rather, it captures settings in which harmful conduct, once initiated, tends to persist because it becomes normalized, socially reinforced, or costly to reverse \cite{ashforth2003, vaughan1996}.  A reversible extension is an important generalization, but it belongs to a different modeling class, nonprogressive threshold dynamics, and is left for future work (Section~\ref{sec:future}).

To state these rules formally, for each non-seed agent $i \notin S_0$, the update rule is:
\begin{equation}
\label{eq:dynamics}
x_i(t+1) = \max\!\Big\{x_i(t),\;\mathbf{1}\!\left[\sigma_i + \sum_{j=1}^n w_{ij}\, x_j(t) \;\geq\; \tau_i\right]\!\Big\},
\end{equation}
where $\mathbf{1}[\cdot]$ is the indicator function, equal to $1$ when its argument is true and $0$ otherwise.  Reading equation~\eqref{eq:dynamics} from the inside out: the sum $\sum_j w_{ij} x_j(t)$ totals the social influence pressure on agent $i$ from all currently active peers.  Adding the baseline pressure $\sigma_i$ gives the total pressure.  The indicator checks whether this total meets or exceeds the threshold $\tau_i$.  The outer $\max$ ensures that an agent who has already activated stays active.

Because $\tau_i = \theta_i + \sigma_i$, the baseline pressure $\sigma_i$ cancels from both sides of the activation inequality.  In terms of net thresholds, the activation condition for an inactive agent simplifies to
\begin{equation}
\label{eq:activation}
\sum_{j=1}^n w_{ij}\, x_j(t) \;\geq\; \theta_i:
\end{equation}
agent $i$ activates when the social influence pressure from active peers alone meets or exceeds $\theta_i$.  We write $w_i(t) = \sum_{j=1}^n w_{ij}\, x_j(t)$ for the total social pressure on agent $i$ at time $t$.

Because agents can only switch from inactive to active and never the reverse, the system is monotone: the set of active agents can only grow.  If no new agent activates at some time step, the state is unchanged.  Since the update rule depends only on the current state, no new agent will activate at any future step either---the system has reached a fixed point.  We write $x^*(S_0) = (x^*_1, \ldots, x^*_n)$ for this final state vector, where $x^*_i = 1$ if agent $i$ is active when the system stops and $x^*_i = 0$ otherwise.  Subsequently we drop the $S_0$ dependence and write simply $x^*$ when the seed set is clear from context.
Since there are only $n$ agents, at most $n$ new activations can occur before the system must stall, so a fixed point is always reached.
\begin{definition}[Cascade]
\label{def:cascade}
A \emph{cascade} from seed set $S_0$ occurs if $\sum_i x_i^* > |S_0|$: at least one non-seed agent has been activated.
\end{definition}

The present model tracks when institutional members become active in harmful role behavior under baseline pressure, peer influence, and exogenous seeding.  A full model of structural harm would require a second layer mapping internal activation to downstream outcomes for external populations---clients, communities, the public.  Building that second layer is an important future direction (Section~\ref{sec:future}).

\section{Cascade Analysis}
\label{sec:cascade}

Our first question is whether a single externally activated agent can trigger a cascade.
We start there, extend to seeds of bounded size, and then develop a mean-field approximation for large populations.

\subsection{Single-seed cascades}

Suppose $S_0 = \{j\}$, a single seed.  At $t = 0$, only agent $j$ is active: $x_j(0) = 1$ and $x_\ell(0) = 0$ for all $\ell \neq j$.  At $t = 1$, the social influence pressure on each non-seed agent $i \neq j$ is
\begin{equation}
\sum_{\ell=1}^n w_{i\ell}\, x_\ell(0) = w_{ij} \cdot 1 + \sum_{\ell \neq j} w_{i\ell} \cdot 0 = w_{ij}.
\end{equation}
So agent $i$ activates at $t = 1$ if and only if $w_{ij} \geq \theta_i$.

\begin{proposition}[No-cascade condition, single seed]
\label{prop:main}
Under separation, no cascade occurs from seed $\{j\}$ if and only if $w_{ij} < \theta_i$ for all $i \neq j$.  A sufficient condition for no cascade from \emph{any} single seed is that the largest off-diagonal entry of $W$ satisfies $\max_{i, j \in V,\, i \neq j} w_{ij} < \underline{\theta}$, where $\underline{\theta} = \min_i \theta_i$.
\end{proposition}
\begin{proof}
The proposition makes two claims.  The first is a necessary and sufficient condition for a specific seed $\{j\}$; the second is a sufficient condition that works for every single seed simultaneously.

\emph{First claim.}  Suppose $w_{ij} < \theta_i$ for all $i \neq j$; we show no cascade occurs.
At $t = 0$ only the seed $j$ is active, so the pressure on each non-seed agent $i$ during the first update is $\sum_{\ell} w_{i\ell} x_\ell(0) = w_{ij} < \theta_i$.  Hence no non-seed activates at $t = 1$: $x(1) = x(0)$.  By monotonicity, $x(t) = x(0)$ for all $t \geq 0$, so $x^* = x(0)$ and no cascade occurs.

Conversely, if $w_{ij} \geq \theta_i$ for some $i \neq j$, then agent $i$ activates at $t = 1$, which constitutes a cascade.

\emph{Second claim.}  Suppose $\max_{i, j \in V,\, i \neq j} w_{ij} < \underline{\theta} = \min_i \theta_i$.  Then for every possible seed $j$ and every non-seed agent $i$, we have $w_{ij} \leq \max_{i', j' \in V,\, i' \neq j'} w_{i'j'} < \underline{\theta} \leq \theta_i$, so the first claim applies and no cascade occurs from any single seed.
\end{proof}

As a special case, if all influence weights are bounded by a common value $\bar{w}$---that is, $w_{ij} \leq \bar{w}$ for all $i \neq j$---and $\bar{w} < \underline{\theta}$, then no single-agent seed triggers a cascade.

\subsection{Multi-seed cascades}
\label{sec:topk}

For each agent $i$ and positive integer $k \leq n-1$, define the \emph{top-$k$ incoming influence sum}:
\begin{equation}
T_k(i) = \sum_{\ell=1}^{k} w_{i(\ell)},
\end{equation}
where $w_{i(1)} \geq w_{i(2)} \geq \cdots \geq w_{i(n-1)}$ are the entries $\{w_{ij} : j \neq i\}$ sorted in nonincreasing order.  For example, if row $i$ of $W$ has off-diagonal entries $(0.9, 0.5, 0.2, 0.1)$, then $T_1(i) = 0.9$, $T_2(i) = 0.9 + 0.5 = 1.4$, and $T_3(i) = 0.9 + 0.5 + 0.2 = 1.6$.

\begin{definition}[$k$-robustness]
\label{def:k-robust}
An institution is \emph{$k$-robust} if no cascade occurs from any seed set $S_0 \subseteq V$ with $|S_0| \leq k$.
\end{definition}

\begin{theorem}[Exact robust $k$-seed criterion]%
\label{thm:topk}
Let $k \leq n-1$ be a positive integer.  Under separation, no cascade occurs from any seed set $S_0$ with $|S_0| \leq k$ if and only if $T_k(i) < \theta_i$ for every $i \in V$.
\end{theorem}

\begin{proof}
Suppose first that $T_k(i) < \theta_i$ for every $i$; we show no cascade occurs from any seed set of size $\leq k$.  Let $S_0$ be any seed set with $|S_0| \leq k$.  At $t = 1$, the social influence pressure on any non-seed agent $i$ is
\begin{equation}
\sum_{j \in S_0} w_{ij}.
\end{equation}
This is a sum of at most $k$ entries from row $i$ of $W$.  Since $T_k(i)$ is defined as the sum of the $k$ \emph{largest} entries in row $i$ (excluding the diagonal), any subset of at most $k$ entries sums to at most $T_k(i)$.  Therefore
\begin{equation}
\sum_{j \in S_0} w_{ij} \;\leq\; T_k(i) \;<\; \theta_i,
\end{equation}
so agent $i$ does not activate at $t = 1$.  Since this holds for every non-seed agent $i$, we have $x(1) = x(0)$, and by monotonicity no cascade occurs.

Conversely, if $T_k(i_0) \geq \theta_{i_0}$ for some agent $i_0$, we exhibit a seed set of size $k$ that activates $i_0$.  Let $S_0$ consist of the $k$ agents $j \neq i_0$ with the largest values of $w_{i_0 j}$ (breaking ties arbitrarily).  By definition, $\sum_{j \in S_0} w_{i_0 j} = T_k(i_0) \geq \theta_{i_0}$, so $i_0$ activates at $t = 1$, which constitutes a cascade.
\end{proof}

To build intuition for the theorem, consider what the condition $T_k(i) < \theta_i$ says about a specific agent.  Agent $i$'s net threshold $\theta_i$ is the amount of social pressure needed to push them into harm.  The quantity $T_k(i)$ is the worst case: the maximum social pressure that any group of $k$ peers could jointly exert on agent $i$.  The theorem says the institution is safe against any $k$-agent shock if and only if, for every agent, even the worst-case coalition of $k$ peers cannot generate enough pressure to cross the threshold.

Two limiting cases illustrate the result.  In a flat organization where all influence weights are approximately equal, $w_{ij} \approx w$, the top-$k$ sum is approximately $T_k(i) \approx kw$.  The condition $T_k(i) < \theta_i$ then becomes $kw < \underline{\theta}$, where $\underline{\theta} = \min_i \theta_i$ is the smallest net threshold in the population.  The institution can therefore withstand any shock of size $k$ as long as $k < \underline{\theta}/w$.  The ratio of the weakest agent's safety margin to the common influence weight sets the robustness limit.  In a hierarchy, by contrast, some agents receive large influence from several powerful sources, making their $T_k(i)$ values much larger than the population average.  These agents are the binding constraint: the institution's vulnerability is determined not by the typical agent but by whoever faces the most concentrated incoming influence.

\subsection{Mean-field approximation}
\label{sec:mean-field}

Theorem~\ref{thm:topk} gives an exact answer for any finite institution.  But applying it requires knowing every entry of the influence matrix $W$, a level of detail often unavailable
for large organizations and one that obscures the qualitative picture.  We instead develop an approximation that operates at a different level of description.  We trade exactness for tractability by developing a \emph{mean-field approximation}, a simplified model that replaces the full network with a single aggregate
measure of social influence. The payoff is a scalar quantity $\lambda_0$ that plays the same role for institutional harm that the basic reproduction number $R_0$ plays in epidemiology.  When $\lambda_0 > 1$, the harm-free state is unstable and a cascade can grow from an arbitrarily small shock.  When $\lambda_0 < 1$, small shocks stay small, producing a harm level comparable to the seed fraction rather than an amplified cascade.

The mean-field model is formulated directly in terms of aggregate quantities---the coupling strength $\alpha$, the threshold distribution $F$, and the harm rate $H$---rather than derived from the finite model via a scaling limit.
This is standard practice in mathematical sociology and statistical physics, where agent-based models and mean-field models routinely coexist as complementary descriptions of the same phenomenon at different scales \cite{granovetter1978, watts2002}.

Suppose $n$ is large.  Let $p = |S_0|/n$ be the fraction of agents who are seeds.  Let $H(t) \in [0,1]$ denote the \emph{harm rate}: the fraction of all agents who are active at time $t$.
The approximation rests on two simplifying assumptions.

\begin{assumption}[Aggregate pressure]
\label{ass:aggregate-pressure}
In the mean-field approximation, the social pressure on each non-seed agent $i$ at time $t$ is the aggregate
\begin{equation}
\alpha \, H(t),
\end{equation}
where $\alpha > 0$ is a coupling constant measuring the overall strength of social influence in the institution and $H(t) \in [0,1]$ is the harm rate at time $t$.
\end{assumption}

In the exact model, the social pressure on agent $i$ is $\sum_j w_{ij} x_j$, a sum that depends on the full influence matrix and on exactly which agents are active.  In a large institution where no single agent dominates, this pressure depends less on \emph{who} is active and more on the overall \emph{proportion}
of active agents.  Assumption~\ref{ass:aggregate-pressure} replaces the agent-specific sum with the single aggregate quantity $\alpha H(t)$.

Three concrete readings make $\alpha$ less abstract.  At full activation ($H = 1$), the pressure on each agent equals $\alpha$, so $\alpha$ is the maximum possible peer pressure any non-active agent can feel.  In the uniform complete-graph case where $w_{ij} = w$ for $i \neq j$, the social pressure simplifies to $w \cdot \#\{\text{active}\} = w n H(t)$, identifying $\alpha = w n$ as the total peer influence each agent receives when all peers are active.  And in the mean-field, the pair $(\alpha, F)$ replaces the institutional structure $\mathcal{S} = (\sigma, W)$ from Definition~\ref{def:structure}.  Here $\alpha$ summarizes how the institution couples its members, $F$ summarizes the population's resistance distribution, and the standing pressure $\sigma$ has been folded into the net thresholds $\theta_i = \tau_i - \sigma_i$.

\begin{assumption}[Random thresholds]
\label{ass:random-thresholds}
The net thresholds $\theta_1, \ldots, \theta_n$ are independent copies of a single random variable $\theta$ with cumulative distribution function $F(x) = P(\theta \leq x)$.  The distribution has no mass at or below zero, $F(0) = 0$, so that separation transfers to the mean-field setting.  The density $f = F'$ is continuous on $(0, \infty)$ with finite right-hand limit at zero,
\begin{equation}
f(0^+) \;=\; \lim_{x \to 0^+} f(x).
\end{equation}
\end{assumption}

Rather than specifying each net threshold $\theta_i$ individually,
Assumption~\ref{ass:random-thresholds} treats them as iid draws from a common distribution.  $F(x)$ is the proportion of agents whose net threshold falls at or below $x$.
The quantity $f(0^+)$ from Assumption~\ref{ass:random-thresholds} has a concrete interpretation: it measures the density of agents whose net thresholds are near zero.  These are agents almost ready to be pushed into harmful behavior by even a small amount of social pressure.  Assumption~\ref{ass:random-thresholds} also makes $F$ continuous on $[0, \infty)$, since differentiability gives continuity on $(0, \infty)$ and right-continuity of the CDF together with $F(0) = 0$ gives continuity at zero.

Under the first assumption,
the activation condition $w_i(t) \geq \theta_i$ from the exact model simplifies to $\alpha H \geq \theta_i$, or equivalently $\theta_i \leq \alpha H$.  Each individual agent's activation is still deterministic: agent $i$ activates if and only if $\theta_i \leq \alpha H$.
What probability contributes is a \emph{counting argument}.  Out of all non-seed agents, the proportion whose thresholds fall at or below $\alpha H$ is $\#\{i \notin S_0 : \theta_i \leq \alpha H\}/(n - |S_0|)$.  Because the $\theta_i$ are independent draws from $F$, the law of large numbers guarantees that this proportion converges to $F(\alpha H)$
as $n$ grows.  So approximately a proportion $F(\alpha H)$ of non-seeds activate.
Adding the seed fraction $p$, the next harm rate is
\begin{equation}
\label{eq:update}
H(t+1) \;:=\; p + (1-p)\,F\bigl(\alpha\, H(t)\bigr).
\end{equation}

Equation~\eqref{eq:update} switches number systems.  In the exact finite model of Section~\ref{sec:dynamics}, $H(t) = \sum_i x_i(t)/n$ is rational by construction: it is a count of active agents divided by $n$, so $nH(t) \in \mathbb{Z}$ at every time step.  On the right-hand side of~\eqref{eq:update}, by contrast, the CDF $F$ is a real-valued function and $F(\alpha H(t)) \in [0,1] \subset \mathbb{R}$, so the update produces a real-valued $H(t+1)$ even when $H(t)$ is rational.  The mean-field model is therefore a real-valued dynamical system on $[0,1]$ that shares its domain of \emph{interpretation} with the exact model---both quantities represent the proportion of a population engaging in harmful behavior---but not its domain of \emph{values}.  Two approximation layers are at work.  First, the empirical CDF $\widehat{F}_n(x) = |\{i : \theta_i \leq x\}|/n$ (rational, piecewise constant) is replaced by the population CDF $F(x)$ (real, continuous); the law of large numbers controls this step.  Second, the literal count of non-seed activations $\lfloor (n - |S_0|)\, F(\alpha H) \rfloor$ is identified with the real number $(n - |S_0|)\, F(\alpha H)$, and we stop rounding.  From~\eqref{eq:update} onward, $H$, $g_p$, $H^*$, and $\lambda_0$ all live in $\mathbb{R}$.

Equation~\eqref{eq:update} gives the next harm rate from the current harm rate.  Fixed points and stability determine which harm rates the institution can settle into.  Define the \emph{mean-field update map}
\begin{equation}
g_p(h) \;:=\; p + (1-p)\,F(\alpha h).
\end{equation} 
When $p$ is clear from context we write simply $g$.
The map $g$ takes a candidate harm rate $h$ and returns the harm rate that the counting argument produces: the seed fraction $p$ plus the proportion of non-seeds whose thresholds are crossed at pressure level $\alpha h$.  A \emph{fixed point} is a self-consistent harm rate, a value $H^*$ satisfying
\begin{equation}
\label{eq:fixedpoint}
H^* = g(H^*) = \underbrace{p}_{\text{seeds}} + \underbrace{(1-p) \cdot F(\alpha H^*)}_{\text{activated non-seeds}}.
\end{equation}
If the harm rate ever reaches $H^*$, the update rule
reproduces exactly the same value.

Recall that a fixed point is \emph{locally stable} if, after a small perturbation, the dynamics return to it, and \emph{unstable} if small perturbations push the system away toward a different steady state.

\begin{proposition}[Mean-field cascade threshold]
\label{prop:scalar}
In the mean-field approximation:
\begin{enumerate}[label=(\roman*)]
    \item For any seed fraction $p \in [0, 1]$, at least one fixed point exists in $[p, 1]$.
    
    \item
    When $p = 0$, the harm-free state $H^* = 0$ is a fixed point.  Define $\lambda_0 = \alpha \cdot f(0^+)$.  When $\lambda_0 > 1$, the harm-free fixed point is unstable and a positive fixed point $H^*_+ > 0$ exists.  When $\lambda_0 < 1$, it is locally stable.
\end{enumerate}
\end{proposition}
\begin{proof}
(i) We verify that $g$ maps $[p, 1]$ into itself.  At the left endpoint: $g(p) = p + (1-p) F(\alpha p) \geq p$ since $F \geq 0$.  At the right endpoint: $g(1) = p + (1-p) F(\alpha) \leq p + (1-p) = 1$ since $F \leq 1$.

Since $g$ is continuous and maps $[p,1]$ into itself, the intermediate value theorem applied to $g(h) - h$---which is $\geq 0$ at $h = p$ and $\leq 0$ at $h = 1$---gives a fixed point in $[p, 1]$.

(ii)
When $p = 0$, the fixed-point equation becomes $H^* = F(\alpha H^*)$.  Write $g_0(h) = F(\alpha h)$.  Since $F(0) = 0$ under separation, $g_0(0) = 0$, so $H^* = 0$ is always a fixed point.

To determine its stability, we compute $g_0'(0^+)$.  By the chain rule,
\begin{equation}
g_0'(h) = \frac{d}{dh} F(\alpha h) = \alpha \cdot f(\alpha h),
\end{equation}
so $g_0'(0^+) = \alpha \cdot f(0^+) = \lambda_0$.

If $\lambda_0 > 1$: then $g_0'(0^+) > 1$, meaning $g_0$ initially rises faster than the diagonal $h \mapsto h$.  By the mean value theorem, $g_0(h) = g_0'(\xi_h)\, h$ for some $\xi_h \in (0, h)$, and $g_0'(\xi_h) \to \lambda_0 > 1$ as $h \downarrow 0$, so for small $h > 0$ we have $g_0(h) > h$ and the dynamics push the harm rate upward.  But $g_0(1) = F(\alpha) \leq 1$, so $g_0(h) - h$ is positive near $h = 0$ and nonpositive at $h = 1$.  Since $g_0(h) - h$ is continuous and changes sign, the intermediate value theorem guarantees a zero: there exists $H^*_+ > 0$ with $g_0(H^*_+) = H^*_+$.  The harm-free fixed point $H^* = 0$ is unstable because small perturbations grow ($g_0(h) > h$ near 0).

If $\lambda_0 < 1$: then $g_0'(0^+) < 1$, meaning $g_0$ initially rises slower than the diagonal.  By the same argument, $g_0(h) < h$ for small $h > 0$, and perturbations from 0 shrink.  The harm-free fixed point is locally stable.
\end{proof}

The quantity $\lambda_0 = \alpha \cdot f(0^+)$ deserves careful interpretation, because it is the central diagnostic of the mean-field model.  It is the product of two factors.  The first, $\alpha$, measures how strongly the institution transmits social influence: a high $\alpha$ means that each agent's behavior weighs heavily on others.
The second, $f(0^+)$, measures how many agents are near the tipping point---those whose net thresholds are close to zero and who therefore need only a small push to begin engaging in harm.  The product $\lambda_0 = \alpha f(0^+)$ captures the feedback loop.  When one agent activates, the resulting social pressure recruits new agents in proportion to both how strongly pressure is transmitted ($\alpha$) and how many agents are susceptible ($f(0^+)$).

This is precisely the logic behind the basic reproduction number $R_0$ in epidemiology.  There, $R_0$ counts how many new infections a single infected individual generates in a fully susceptible population; when $R_0 > 1$, an epidemic can grow, and when $R_0 < 1$, outbreaks die out.  The cascade threshold $\lambda_0$ plays an analogous role at the level of local stability.  When $\lambda_0 > 1$, the harm-free state is locally unstable and a positive-harm fixed point exists; when $\lambda_0 < 1$, the harm-free state is locally stable.  The analogy is not exact.  Proposition~\ref{prop:scalar} establishes only local stability, not global convergence from arbitrary initial conditions.  An institution designer who can estimate $\alpha$ and the threshold distribution $F$ can therefore compute $\lambda_0$ and determine whether the institution is structurally vulnerable to cascading harm.

The conclusions of Proposition~\ref{prop:scalar}(ii)
are not artifacts of the idealized $p = 0$ limit.  They persist when a small fraction of agents are seeded.  We call an interior fixed point $H^*_+ \in (0, 1)$ of $g_0$ \emph{simple} if the curves $y = g_0(h)$
and $y = h$ cross transversally there, meaning $g_0'(H^*_+) \neq 1$.  Equivalently, the diagonal $y = h$ is not the tangent line to $y = g_0(h)$ at $H^*_+$: geometrically, the update curve slices through the diagonal rather than just touching it.

Before stating the persistence result, a word on the status of $g_0$ itself.  When $p = 0$ and $F(0) = 0$ under separation, the dynamics are trivial: $H(0) = 0$ and $H(t) = 0$ for all $t \geq 0$, so no non-seed agent ever activates.  The positive fixed points $H^*_+ > 0$ analyzed in Proposition~\ref{prop:scalar}(ii) are therefore not trajectories the system visits at $p = 0$.  They are features of the abstract map $g_0 \colon [0,1] \to [0,1]$, whose role is to organize the behavior of the nearby maps $g_p$ with $p > 0$---and those $g_p$ \emph{do} describe genuine physical situations, an institution subjected to an exogenous shock that activates a small fraction $p$ of agents.  The limit $p \downarrow 0$ is accordingly a thought experiment about the harm-free state's response to vanishingly small shocks, not a process in which any particular institution changes over time.  The corollary below makes the reference-map role precise: the fixed-point structure of $g_0$ persists under perturbation to $g_p$ for small $p > 0$.

\begin{corollary}[Persistence at small positive seed]%
\label{cor:persist}
If $\lambda_0 > 1$ and $H^*_+$ is a simple fixed point of $g_0$, then for small $p > 0$ a fixed point of $g_p$ exists near $H^*_+$.  If $\lambda_0 < 1$, fix any $\rho > 0$ such that $g_0(h) < h$ for all $h \in (0, \rho]$; such a $\rho$ exists because $g_0'(0^+) = \lambda_0 < 1$.  Then every fixed point $H^*(p)$ of $g_p$ in $(0, \rho]$ satisfies $H^*(p) \to 0$ as $p \downarrow 0$.
\end{corollary}

When $\lambda_0 > 1$, even an infinitesimally small seed group sustains a positive level of harm.  The cascade equilibrium does not vanish just because the initial shock is tiny.  Conversely, when $\lambda_0 < 1$, the harm rate shrinks to zero as the seed fraction shrinks, confirming that the institution genuinely absorbs small shocks. 

\begin{proof}
The strategy is to show that the maps $g_p$ (with seeds) are close to the map $g_0$ (without seeds), so the fixed-point structure cannot change abruptly.

Write $g_p(h) = p + (1-p)F(\alpha h)$.  Since $|g_p(h) - g_0(h)| = p|1 - F(\alpha h)| \leq p$, the family $g_p$ converges uniformly to $g_0$ on $[0,1]$ as $p \downarrow 0$.
Informally, adding a tiny seed fraction shifts the update map by a tiny amount.

\emph{First claim.}  A transversal crossing of the curve $y = g_0(h)$ with the diagonal $y = h$ is structurally stable: a small perturbation of the curve cannot eliminate it.  Formally, because $g_0'(H^*_+) \neq 1$, the function $g_0(h) - h$ changes sign at $H^*_+$: there exists $\eta > 0$ small enough that $H^*_+ \pm \eta \in (0, 1)$ and the values $g_0(H^*_+ - \eta) - (H^*_+ - \eta)$ and $g_0(H^*_+ + \eta) - (H^*_+ + \eta)$ have opposite signs.  Let
\begin{equation}
\varepsilon_0 \;=\; \min\Bigl\{\bigl|g_0(H^*_+ - \eta) - (H^*_+ - \eta)\bigr|,\; \bigl|g_0(H^*_+ + \eta) - (H^*_+ + \eta)\bigr|\Bigr\},
\end{equation}
the smaller of the two distances from $g_0$ to the diagonal at the test points.  By uniform convergence, there exists $p_0 > 0$ such that $|g_p(h) - g_0(h)| < \varepsilon_0$ for all $h \in [0, 1]$ and all $p < p_0$.  In particular, $g_p(H^*_+ \pm \eta) - (H^*_+ \pm \eta)$ has the same sign as $g_0(H^*_+ \pm \eta) - (H^*_+ \pm \eta)$ for $p < p_0$.  So $g_p(h) - h$ is continuous on $[H^*_+ - \eta, H^*_+ + \eta]$ and changes sign on this interval.
By the intermediate value theorem, $g_p$ has a fixed point in $(H^*_+ - \eta, H^*_+ + \eta)$.

\emph{Second claim.}  Let $\rho$ be as in the statement and fix $\varepsilon \in (0, \rho)$.  The function $h - g_0(h)$ is continuous and strictly positive on the compact interval $[\varepsilon, \rho]$, so $c(\varepsilon) := \min_{h \in [\varepsilon, \rho]} \bigl(h - g_0(h)\bigr) > 0$.  For $p < c(\varepsilon)$, uniform convergence gives $g_p(h) \leq g_0(h) + p < h$ for all $h \in [\varepsilon, \rho]$, so $g_p$ has no fixed point in $[\varepsilon, \rho]$.  Any fixed point of $g_p$ in $(0, \rho]$ therefore satisfies $H^*(p) < \varepsilon$, and since $\varepsilon \in (0, \rho)$ was arbitrary, $H^*(p) \to 0$ as $p \downarrow 0$.
\end{proof}

\section{Signed Influence}
\label{sec:signed}

Up to this point, every entry of the influence matrix $W$ has been nonnegative---each peer relationship pushes a coworker toward engaging in the harmful behavior, never away from it.  But not all peer influence works in the same direction.  When one agent's misconduct draws sanctions, triggers oversight that reaches their peers, or provokes peer disapproval rather than emulation, the response can \emph{reduce} pressure on those nearby rather than amplify it.  We now extend the model to accommodate both kinds of peer influence, generalize the robustness criterion from Section~\ref{sec:topk}, and establish a convexity property of worst-case exposure that will be essential for the repair theory in Section~\ref{sec:repair}.

\subsection{The signed influence matrix}

Replace the nonnegative influence matrix $W$ with a \emph{signed influence matrix} $\mathcal{A} \in \mathbb{R}^{n \times n}$, retaining the zero-diagonal convention $a_{ii} = 0$.  For each pair of agents with $i \neq j$, let $w^+_{ij} \geq 0$ be the harm-promoting influence of $j$ on $i$ and $w^-_{ij} \geq 0$ be the protective influence of $j$ on $i$, so that
\begin{equation}
a_{ij} \;=\; w^+_{ij} - w^-_{ij}.
\end{equation}
The net entry $a_{ij}$ represents the \emph{total institutional effect} of $j$'s activation on $i$'s pressure, after accounting for both direct contagion and any induced countervailing mechanisms.  A negative entry $a_{ij} < 0$ does not mean that harmful behavior is intrinsically protective.  It means that once all institutional responses are accounted for---oversight triggered by $j$'s misconduct, sanctions activated by $j$'s visibility, heightened caution among $i$'s peers---the net effect of $j$'s activation on $i$'s pressure is protective.

The signed activation rule is that non-seed agent $i$ activates at time $t+1$ if not already active and
\begin{equation}
\label{eq:signed-dynamics}
\sigma_i + \sum_{j=1}^n a_{ij}\, x_j(t) \;\geq\; \tau_i, \qquad\text{equivalently}\qquad \sum_{j=1}^n a_{ij}\, x_j(t) \;\geq\; \theta_i.
\end{equation}
Activation remains absorbing, so once $x_i = 1$, it stays 1.  In the nonnegative model, the dynamics are \emph{monotone} in the sense that activating more agents can never reduce anyone's pressure.  Signed influence destroys this property.  Activating an agent $j$ with $a_{ij} < 0$ (net protective influence on $i$) \emph{decreases} pressure on $i$.  However, signed influence preserves the property that makes our exact theorem work.  If no non-seed agent activates at $t = 1$, the system has reached its fixed point.  The reasoning is the same induction used in Proposition~\ref{prop:main}.  The base case is $x(1) = x(0)$ (no non-seed activated).  For the inductive step, if $x(t) = x(0)$ then every agent sees the same state at time $t$ as at time $0$, so the update rule gives $x(t+1) = x(t) = x(0)$.  First-step prevention is therefore still sufficient to prevent all cascades.

\subsection{The signed robust \texorpdfstring{$k$}{k}-seed criterion}

For each agent $i$, define the \emph{positive part} of the signed influence:
\begin{equation}
a^+_{ij} \;=\; \max\{a_{ij},\, 0\}.
\end{equation}
This is the net harm-promoting influence of $j$ on $i$, after accounting for any protective effect.  Define the \emph{signed top-$k$ incoming influence sum}:
\begin{equation}
\widetilde{T}_k(i) = \sum_{\ell=1}^{k} a^+_{i(\ell)},
\end{equation}
where $a^+_{i(1)} \geq a^+_{i(2)} \geq \cdots$ are the positive parts $a^+_{ij}$ sorted in nonincreasing order.

\begin{theorem}[Signed exact robust $k$-seed criterion]
\label{thm:signed}
Let $1 \leq k \leq n-1$.  Under the separation condition, no cascade occurs from any seed set $S_0$ with $|S_0| \leq k$ if and only if $\widetilde{T}_k(i) < \theta_i$ for every $i \in V$.
\end{theorem}

\begin{proof}
First we show that if $\widetilde{T}_k(i) < \theta_i$ for every $i$, then no cascade occurs from any seed set of size $\leq k$.
Under this hypothesis, the proof of Theorem~\ref{thm:topk} carries over after replacing $w_{ij}$ with $a^+_{ij}$.  For any seed set $S_0$ with $|S_0| \leq k$ and any non-seed $i$, the first-step pressure on $i$ satisfies
\[
\sum_{j \in S_0} a_{ij} \;\leq\; \sum_{j \in S_0} a^+_{ij} \;\leq\; \widetilde{T}_{|S_0|}(i) \;\leq\; \widetilde{T}_k(i) \;<\; \theta_i,
\]
using (i) $a_{ij} \leq a^+_{ij}$, (ii) the fact that $\widetilde T_{|S_0|}(i)$ is the maximum sum of $|S_0|$ positive parts of row $i$, (iii) $|S_0| \leq k$ so the top-$|S_0|$ sum is bounded by the top-$k$ sum, and (iv) the hypothesis.  Hence no non-seed activates at $t = 1$, so $x(1) = x(0)$, and the same induction as in Theorem~\ref{thm:topk} extends this to all $t$.

Conversely, suppose $\widetilde{T}_k(i_0) \geq \theta_{i_0}$ for some agent $i_0$.  Let $S^\star$ contain the agents corresponding to the strictly positive entries among the $k$ largest positive parts of row $i_0$; then $|S^\star| \leq k$, and the omitted top-$k$ terms are zero.  Hence
\[
\sum_{j \in S^\star} a_{i_0 j} \;=\; \widetilde{T}_k(i_0) \;\geq\; \theta_{i_0},
\]
so $i_0$ activates at $t = 1$.
\end{proof}

The signed criterion has a natural institutional interpretation.  Protective influence enters through the netting operation.  If agent $j$ exerts both harm-promoting influence $w^+_{ij}$ and protective influence $w^-_{ij}$, only the positive part $a^+_{ij} = \max\{w^+_{ij} - w^-_{ij},\, 0\}$ feeds into $\widetilde T_k(i)$.  A peer whose protective effect on $i$ exceeds their harm-promoting effect has $a^+_{ij} = 0$ and contributes nothing to the worst-case pressure.  This captures a real limitation of protective influence.  It reduces exposure when seeds are random, but it cannot improve safety margins in the worst case, because the worst-case seed set leaves out the protective agents.

\subsection{Convexity of signed influence sums}
\label{sec:signed-convexity}

Theorem~\ref{thm:signed} tells us \emph{whether} an institution is safe under signed influence, but it does not tell us how to \emph{make} it safe.  In Section~\ref{sec:repair} we will formulate repair as an optimization problem---minimizing the cost of interventions that restore safety.  The tractability of that optimization depends on a structural property of the signed top-$k$ sum, which we establish here.  It is a \emph{convex} function.  Convexity matters because convex optimization problems have no spurious local minima, so off-the-shelf convex solvers can compute the globally optimal repair.

Recall from Theorem~\ref{thm:signed} that to determine whether agent $i$ is safe, we look at the $k$ largest \emph{positive} influences on $i$ and add them up.  Negative influences are ignored because they do not contribute to the worst-case pressure on $i$.  Any seed set seeking to maximize pressure on $i$ would skip agents with $a_{ij} \leq 0$, since such a seed adds zero or actively reduces the total.  The signed top-$k$ sum packages this operation into a single function.

For any vector $v = (v_1, \ldots, v_m)$, write $v^+ = \max\{v, 0\}$ for the componentwise positive part (replace each negative entry with zero), and let $v^+_{(1)} \geq v^+_{(2)} \geq \cdots$ be the positive parts sorted from largest to smallest.  Define
\begin{equation}
\Phi_k(v) \;=\; \sum_{\ell=1}^{k} v^+_{(\ell)},
\end{equation}
the sum of the $k$ largest positive components.  For example, if $v = (3, -1, 2, 0)$ and $k = 2$, the positive parts are $(3, 0, 2, 0)$, the two largest are $3$ and $2$, so $\Phi_2(v) = 5$.  The negative entry $-1$ is ignored entirely.

In our model, $v$ is the off-diagonal part of row $i$ of the signed influence matrix: the vector $(a_{ij})_{j \neq i}$ of all influences on agent $i$ from other agents.  So $\Phi_k$ applied to this vector gives $\widetilde{T}_k(i)$, the worst-case first-step pressure from Theorem~\ref{thm:signed}.  The function $\Phi_k$ generalizes $T_k$ from Section~\ref{sec:topk} to signed vectors.  Appending zero or nonpositive components to the argument does not change $\Phi_k$, since only positive parts enter the sum, so applying $\Phi_k$ to the off-diagonal part of a row or to the full row gives the same value.

\begin{proposition}[$\Phi_k$ is convex]
\label{prop:phik-convex}
For $1 \leq k \leq m$, the function $\Phi_k \colon \mathbb{R}^m \to \mathbb{R}$ is convex.
\end{proposition}

\begin{proof}
The function $\Phi_k$ admits a finite combinatorial maximum representation:
\begin{equation}
\Phi_k(v) \;=\; \max\left\{\sum_{j \in S} v_j \;:\; S \subseteq \{1, \ldots, m\},\; |S| \leq k\right\}.
\end{equation}
The maximum on the right is attained by choosing $S$ to be the set of indices of the (at most $k$) largest positive components of $v$; indices with $v_j \leq 0$ are excluded because including them cannot increase the sum.  The right-hand side is a pointwise maximum of finitely many linear functions of $v$, one for each admissible subset $S$.  A pointwise maximum of convex functions is convex \cite{boyd2004}, and linear functions are convex, so $\Phi_k$ is convex.
\end{proof}

This convexity will play a key role when we formulate repair problems in Section~\ref{sec:repair}.

\section{Fairness and Impossibility}
\label{sec:impossibility}

The cascade framework so far has been group-neutral.
 It treats all agents as individuals, without asking whether the institution pushes different groups into harmful behavior at different rates.  In practice, not everyone faces the same baseline pressure or has the same net threshold.  Junior employees may face higher baselines than senior ones; frontline workers may be more susceptible to peer pressure than administrators.  The question is whether a single institutional design---one set of rules, one coordination structure, applied uniformly---can keep \emph{all} groups from being cascaded into harm.

This section shows that the answer is often \emph{no}.  We partition the institution into two groups, one more easily activated and one less so, and prove two results.  First, an exact finite-agent impossibility theorem (Section~\ref{sec:functionality}) shows that coordination requirements can cascade the more exposed group into harmful behavior while posing no threat to the more insulated group.  Second, a mean-field analysis (Section~\ref{sec:mean-field-fairness}) shows that when members are more influenced by their own group than by the other group---a phenomenon called \emph{homophily}---the institution can tip from safe to unsafe.  The destabilizing effect of homophily is monotone.  Every increase in within-group affinity moves the institution closer to a cascade, never further from one.

\subsection{Finite-agent impossibility}
\label{sec:functionality}

A central function of institutions is to coordinate their members, directing collective action, enforcing norms, and ensuring that individuals respond to one another \cite{march1958}.  For coordination to work, members must be responsive to peer influence---if no one listens to anyone else, the institution cannot function.  Mathematically, this means the influence matrix $W$ must have enough nonzero entries of sufficient size.  Each agent must be reachable by a small group of peers who can collectively exert enough pressure to direct them.

A key modeling assumption underlies what follows.  The influence channels that enable coordination are the \emph{same} channels through which harm spreads.  The weight $w_{ij}$ measures how much agent $j$'s behavior affects agent $i$---and this single channel serves both purposes.  When $j$ is doing their job well, the influence helps coordinate $i$; when $j$ is engaging in harmful behavior, the same influence pressures $i$ toward harm.  This assumption draws further support from the empirical settings of Section~\ref{sec:model}.  We introduced them to motivate the prosocial default, but they also illustrate the dual-use nature of influence channels.  In Milgram's experiments, the authority pressure that ordinarily coordinates institutional action was the same pressure that pushed subjects into harmful compliance.  In Browning's account, the command structure that organized the battalion's logistics was the same structure that transmitted pressure to participate in killings.

If coordination and harm-spreading used completely separate influence channels, there would be no tension between functionality and safety.  The impossibility we prove below arises precisely because the channels are shared, and the more responsive agents are to their peers, the easier it is for both coordination \emph{and} harm to propagate.  The question is whether an institution can be functional and safe at the same time.  Throughout this section we work with the original unsigned non-negative influence matrix $W$ from Section~\ref{sec:dynamics}; the signed extension of Section~\ref{sec:signed} is not used here.

To formalize this, we divide the agents into two groups, $G_A$ and $G_B$.  The labels are abstract.  $G_A$ and $G_B$ could represent senior versus junior employees, tenured versus untenured faculty, or any other partition.  Write
\begin{equation}
\underline{\theta}_{G_A} = \min_{i \in G_A} \theta_i, \qquad \underline{\theta}_{G_B} = \min_{i \in G_B} \theta_i
\end{equation}
for the net threshold of each group's most vulnerable member.  We assume $\underline{\theta}_{G_B} \leq \underline{\theta}_{G_A}$: the most easily activated member of $G_B$ is at least as vulnerable as the most easily activated member of $G_A$.  This is what we mean by calling $G_B$ the \emph{more exposed} group and $G_A$ the \emph{more insulated}
group.  The two groups may have overlapping threshold distributions; the distinction rests entirely on the comparison of their minimums.  In particular, the institution's globally most vulnerable agent sits in $G_B$ by definition, but $G_A$ may contain many agents with lower thresholds than many $G_B$ agents.

We formalize ``functional enough to coordinate'' as a minimum responsiveness requirement.  For the institution to direct any of its members, there must exist a small group of peers who can collectively exert enough influence on that member.

\begin{definition}[Minimum coordination requirement]
\label{def:coordination}
An institution satisfies the \emph{aggregate mobilization condition} $\mathcal{G}(r, \Lambda)$ if for every agent $i \in V$, there exists a \emph{coordination set} $C_i \subseteq V \setminus \{i\}$ with $|C_i| \leq r$ such that
\begin{equation}
\sum_{j \in C_i} w_{ij} \geq \Lambda.
\end{equation}
\end{definition}

In words, every agent can be directed by at most $r$ peers exerting combined influence at least $\Lambda$.  The parameter $\Lambda$ measures the \emph{strength} of coordination (how much pressure can be brought to bear) and $r$ measures its \emph{breadth} (how many peers are needed).  Note that $\mathcal{G}$ is monotonic in both parameters: an institution satisfying $\mathcal{G}(r, \Lambda)$ automatically satisfies $\mathcal{G}(r', \Lambda')$ for any $r' \geq r$ and $\Lambda' \leq \Lambda$, since both relaxations weaken the requirement.  For example, suppose an institution has five agents and, for each agent $i$, there exist two other agents in the institution whose combined influence weights on $i$ sum to at least $0.8$.  Then the institution satisfies $\mathcal{G}(2, 0.8)$.  We will see shortly that if an adversary can seed $k = 1$ agent, the adversary can force at least $(1/2) \cdot 0.8 = 0.4$ units of incoming pressure on every agent in the worst case.  If $k = 2$ or more, the full coordination strength $0.8$ is unavoidable.  These lower bounds are the \emph{exposure floor} that the coordination requirement imposes.

A natural default---and the case most relevant to fairness concerns---is to apply the \emph{same} coordination standard to everyone: the same $r$ and $\Lambda$ for both groups.  We call this a \emph{neutral} or \emph{uniform} design, meaning the institution does not tailor its coordination requirements to account for the fact that group $G_B$ has lower thresholds than group $G_A$.  The question is whether there exists any institutional design, any choice of influence matrix $W$, that simultaneously satisfies the coordination requirement and maintains $k$-robustness.  The answer is exact.

The argument has three steps.  First, we show that any design meeting the coordination requirement $\mathcal{G}(r,\Lambda)$ must expose every agent to at least a certain amount of worst-case pressure---a quantity we call the \emph{exposure floor} $\Gamma_{k,r}(\Lambda)$.  Second, we combine that lower bound with the exact $k$-robustness criterion of Theorem~\ref{thm:topk} to obtain a sharp frontier between functionality and safety.  Third, we apply the frontier to the two-group setting to show that a neutral coordination standard can be safe for the insulated group while being provably unsafe for the exposed group.

By Theorem~\ref{thm:topk}, $k$-robustness is controlled by the quantities $T_k(i)$.  The institution is $k$-robust if and only if $T_k(i) < \theta_i$ for every agent.  So if the coordination requirement imposes any unavoidable lower bound on $T_k(i)$, that lower bound will determine whether coordination and robustness can coexist.  The key observation is that such a lower bound exists.  The coordination condition forces a minimum value of $T_k(i)$ for every agent, regardless of how the institution distributes its influence weights.

\begin{lemma}[Coordination imposes an exposure floor]
\label{lem:exposure-floor}
Suppose the institution satisfies the coordination condition $\mathcal{G}(r, \Lambda)$, with $1 \leq r \leq n-1$.  Then for every agent $i$ and every integer $k$ with $1 \leq k \leq n-1$,%
\begin{equation}
T_k(i) \;\geq\; \Gamma_{k,r}(\Lambda),
\end{equation}
where the \emph{exposure floor} is
\begin{equation}
\label{eq:exposure-floor}
\Gamma_{k,r}(\Lambda) \;=\; \begin{cases} \dfrac{k}{r}\,\Lambda, & k < r, \\[6pt] \Lambda, & k \geq r. \end{cases}
\end{equation}
The bound is sharp---achievable as an equality. %
Equality holds when row $i$ of $W$ has exactly $r$ entries equal to $\Lambda/r$ and the rest equal to zero.
\end{lemma} 

The two cases have a simple interpretation.  When $k \geq r$, an adversary choosing $k$ seeds can include the entire coordination set, so the full coordination strength $\Lambda$ is available as exposure.  When $k < r$, the adversary can only use $k$ of the $r$ coordinating peers; the best the institution can do is spread the coordination load evenly across all $r$ peers, in which case the $k$ largest each contribute $\Lambda/r$, for a total of $(k/r)\Lambda$.

\begin{proof}
Fix an agent $i$.  By the coordination condition, there exists $C_i \subseteq V \setminus \{i\}$ with $|C_i| \leq r$ such that $\sum_{j \in C_i} w_{ij} \geq \Lambda$.  Let $s = |C_i| \leq r$, and write $w_{i(1)} \geq w_{i(2)} \geq \cdots$ for the off-diagonal entries of row $i$ sorted in decreasing order.  The top $s$ entries are at least as large as any chosen $s$-subset, so
\begin{equation}
\sum_{\ell=1}^{s} w_{i(\ell)} \;\geq\; \sum_{j \in C_i} w_{ij} \;\geq\; \Lambda.
\end{equation}

\emph{Case 1: $k \geq s$.}  Then $T_k(i) \geq \sum_{\ell=1}^{s} w_{i(\ell)} \geq \Lambda$, since adding more nonnegative terms cannot decrease the sum.  Since $\Gamma_{k,r}(\Lambda) \leq \Lambda$ by definition, we conclude $T_k(i) \geq \Gamma_{k,r}(\Lambda)$.

\emph{Case 2: $k < s$.}  The idea is that the only way to make the top $k$ entries as small as possible is to spread the coordination load $\Lambda$ as evenly as possible across all $s$ peers, and even then each of the top $k$ entries inherits at least a $(1/s)$-share.  Formally, since the entries are sorted in decreasing order, the average of the first $k$ entries is at least as large as the average of the first $s$ entries; dropping the smallest $s - k$ entries from the average can only raise it.  Therefore
\begin{equation}
\frac{1}{k}\sum_{\ell=1}^{k} w_{i(\ell)} \;\geq\; \frac{1}{s}\sum_{\ell=1}^{s} w_{i(\ell)} \;\geq\; \frac{\Lambda}{s}.
\end{equation}
Multiplying through by $k$ gives $T_k(i) \geq (k/s)\Lambda$.  Since $s \leq r$, we have $k/s \geq k/r$, so $T_k(i) \geq (k/r)\Lambda = \Gamma_{k,r}(\Lambda)$.

\emph{Sharpness.}  Take row $i$ of $W$ to have exactly $r$ entries equal to $\Lambda/r$ and the rest equal to zero.  The coordination condition is satisfied for $i$ (those $r$ entries sum to $\Lambda$), and $T_k(i) = \min(k, r) \cdot (\Lambda/r) = \Gamma_{k,r}(\Lambda)$.
\end{proof}

Applying this lemma row by row to the influence matrix $W$ yields an exact characterization of when coordination and robustness can coexist.

\begin{theorem}[Sharp coordination--robustness frontier]
\label{thm:functionality}
Let $1 \leq k \leq n-1$ and assume $n \geq r + 1$, so that each agent has at least $r$ distinct peers.
There exists an institutional design $W$ with $w_{ij} \geq 0$ for all $i,j$ satisfying both $\mathcal{G}(r, \Lambda)$ and $k$-robustness if and only if the exposure floor stays below the weakest agent's net threshold:
\begin{equation}
\label{eq:frontier}
\Gamma_{k,r}(\Lambda) \;<\; \underline{\theta} \;=\; \min_i \theta_i.
\end{equation}
\end{theorem}

\begin{proof}
\emph{Exposure floor.}  By Lemma~\ref{lem:exposure-floor}, $T_k(i) \geq \Gamma_{k,r}(\Lambda)$ for every agent $i$.

Suppose the institution satisfies both coordination $\mathcal{G}(r, \Lambda)$ and $k$-robustness.  We show $\Gamma_{k,r}(\Lambda) < \underline{\theta}$.  By $k$-robustness and Theorem~\ref{thm:topk}, $T_k(i) < \theta_i$ for all $i$.  Combining with the exposure floor gives $\Gamma_{k,r}(\Lambda) \leq T_k(i) < \theta_i$ for all $i$, hence $\Gamma_{k,r}(\Lambda) < \underline{\theta}$.

Conversely, assume $\Gamma_{k,r}(\Lambda) < \underline{\theta}$.  Construct a design in which each agent $i$ receives influence $\Lambda/r$ from each of $r$ designated peers---for instance, arrange the agents in a cycle and assign each agent the next $r$ agents as its coordination set.  Every other influence weight is zero.  In this construction, each row of $W$ has exactly $r$ nonzero off-diagonal entries, all equal to $\Lambda/r$, so the coordination condition $\mathcal{G}(r, \Lambda)$ holds.  For any set of $k$ seeds, the sum of their influence weights on agent $i$ is at most $\min(k, r) \cdot (\Lambda/r) = \Gamma_{k,r}(\Lambda)$---precisely the exposure floor.  Since $\Gamma_{k,r}(\Lambda) < \underline{\theta} \leq \theta_i$ for every agent $i$, $k$-robustness holds by Theorem~\ref{thm:topk}.
\end{proof}

The frontier condition~\eqref{eq:frontier} draws a sharp line.  On one side, the institution can be both functional and safe; on the other, no design can achieve both.  When $k \geq r$, the condition reduces to $\Lambda < \underline{\theta}$, meaning the coordination strength must stay below the weakest agent's net threshold.  When $k < r$, the condition is the weaker $(k/r)\Lambda < \underline{\theta}$, reflecting the fact that a smaller adversary cannot fully exploit the coordination structure.  In both cases, the frontier compares an unavoidable exposure level, $\Gamma_{k,r}(\Lambda)$, to the relevant threshold floor.

The frontier theorem is group-neutral.  It tells us whether coordination and robustness can coexist for the institution as a whole, using the global minimum threshold $\underline{\theta} = \min_i \theta_i$.  The fairness issue appears when different groups have different minimum thresholds.  Recall that $\underline{\theta}_{G_B} \leq \underline{\theta}_{G_A}$ by construction.  If the exposure floor $\Gamma_{k,r}(\Lambda)$ falls between these two values, the frontier condition is satisfiable for group $G_A$, since $\Gamma_{k,r}(\Lambda) < \underline{\theta}_{G_A}$, but not for group $G_B$, since $\Gamma_{k,r}(\Lambda) \geq \underline{\theta}_{G_B}$.  In other words, the same neutral coordination standard can be safe for the insulated group while being provably unsafe for the exposed group.  Say a design is \emph{$k$-robust for group $G$} if no seed set of size at most $k$ activates any non-seed agent of $G$.

\begin{corollary}[Neutrality gap]
\label{cor:neutrality-gap}
Let $1 \leq k \leq n-1$, suppose a neutral design imposes the same coordination requirement $\mathcal{G}(r, \Lambda)$ on all agents, and assume $|G_A| \geq r + 1$ and $|G_B| \geq r + 1$.  If
\begin{equation}
\underline{\theta}_{G_B} \;\leq\; \Gamma_{k,r}(\Lambda) \;<\; \underline{\theta}_{G_A},
\end{equation}
then no design satisfying that neutral coordination requirement can be $k$-robust for group $G_B$, but there exist designs satisfying that same requirement that are $k$-robust for group $G_A$.
\end{corollary}

\begin{proof}
For group $G_B$: the exposure floor gives $T_k(i) \geq \Gamma_{k,r}(\Lambda) \geq \underline{\theta}_{G_B}$ for every agent.  In particular, the agent $i^* \in G_B$ with $\theta_{i^*} = \underline{\theta}_{G_B}$ has $T_k(i^*) \geq \theta_{i^*}$, so $k$-robustness fails by Theorem~\ref{thm:topk}.

For group $G_A$: since $\Gamma_{k,r}(\Lambda) < \underline{\theta}_{G_A}$, we construct a design in which $G_A$ agents receive influence only from other $G_A$ agents.  Arrange the agents of $G_A$ in a cycle and assign each the next $r$ agents as its coordination set, each with weight $\Lambda/r$; this requires $|G_A| \geq r + 1$.  Give every $G_A$ agent zero influence weight from all $G_B$ agents.  For $G_B$ agents, assign coordination sets within $G_B$ in the same way, which requires $|G_B| \geq r + 1$.

In this design, the coordination requirement $\mathcal{G}(r,\Lambda)$ holds for all agents.  For any $i \in G_A$, the incoming influence comes entirely from $G_A$, so $T_k(i) = \Gamma_{k,r}(\Lambda) < \theta_i$.  Because no $G_B$ agent exerts influence on any $G_A$ agent, cascades originating in $G_B$ cannot propagate into $G_A$.  By Theorem~\ref{thm:topk} applied to $G_A$ in isolation, $k$-robustness holds for group $G_A$.
\end{proof}

The corollary identifies a precise \emph{satisfiability gap}.  The neutral coordination standard is mathematically compatible with safety for the insulated group but mathematically incompatible with safety for the exposed group.  The institution's options are to weaken the coordination requirement $\Lambda$, reduce baseline pressure on the exposed group (raising $\underline{\theta}_{G_B}$), or abandon uniform standards and tailor the design to each group.

\subsection{Mean-field two-group analysis}
\label{sec:mean-field-fairness}

The finite-agent impossibility result in the previous subsection shows that a uniform coordination standard can force one group into harmful behavior while leaving the other resistant.  A complementary question remains: how do group-level differences and social structure interact to destabilize the institution?

To answer this, we extend the mean-field framework of Section~\ref{sec:mean-field} to two groups.  As before, we replace individual agents with continuous distributions, but now each group has its own distribution of net thresholds, and we track how the two groups influence each other.  The results in this subsection are mathematically exact within the mean-field model but describe averaged large-population behavior rather than the fate of any particular agent.

The subsection proceeds in four steps.  First, we specify each group's threshold distribution under a shared baseline.  Second, we derive the linearized stability matrix (the Jacobian) that governs whether the harm-free state can survive small perturbations.  Third, in the simplest interaction structure, we show that no neutral baseline can equalize the two groups' robustness margins.  Fourth, we introduce homophily and prove that increasing within-group influence eventually destabilizes the institution, with an explicit formula for the tipping point.

We model two groups, $A$ and $B$, within the institution, representing any partition into a more insulated and a more exposed group.  The parameter $s \geq 0$ represents a common baseline pressure level applied to both groups under a neutral institutional design.  It is closer in spirit to the baseline pressure $\sigma_i$ of Section~\ref{sec:structure} than to the social coupling constant $\alpha$ from the single-group mean-field of Section~\ref{sec:mean-field}: $s$ is institutional, not interactional.  We carry $s$ as an explicit argument of the threshold distributions because the analyses below compare robustness across different baseline values.

Each group has its own distribution of net thresholds.  We use a parameterized CDF $F_g(x;\, s)$ in which $x$ is the function's variable (threshold value) and $s$ is a parameter (baseline pressure); the semicolon notation marks this distinction.  For each fixed $s$, $F_g(\cdot;\, s)$ is a CDF in $x$; varying $s$ gives a family of CDFs.

\begin{definition}[Two-group threshold distributions]
\label{def:two-group-cdf}
For each group $g \in \{A, B\}$ and each baseline pressure $s \geq 0$, $F_g(\cdot;\, s)$ is a cumulative distribution function on $[0, \infty)$ satisfying
\begin{enumerate}[label=(\roman*)]
  \item $F_g(0;\, s) = 0$ (no agent has a net threshold of zero or below);
  \item $F_g$ has a continuous density $f_g(x;\, s) = \partial F_g(x;\, s) / \partial x$ on $(0, \infty)$;
  \item the density has a well-defined right-hand limit at zero, $f_g(0^+;\, s) = \lim_{x \downarrow 0} f_g(x;\, s)$.
\end{enumerate}
We call $f_g(0^+;\, s)$ the \emph{boundary density} for group $g$ at baseline $s$.  $F_g(x;\, s)$ is interpreted as the proportion of group-$g$ agents whose net threshold is at most $x$ when the shared baseline pressure is $s$.
\end{definition}

The boundary density $f_g(0^+;\,s)$ will turn out to be the key quantity.  It measures how many agents in group $g$ are \emph{just barely} resistant to activation.  A large boundary density means many agents are sitting right at the edge, so even a small push can tip them over.

To make the analysis fully explicit, we choose a specific family of distributions.
Let
\begin{equation}
\label{eq:parametric}
F_g(x;\, s) = 1 - \exp\!\big(-\lambda_g\, e^{\eta s}\, x\big), \qquad x \geq 0.
\end{equation}
This is the CDF of an exponential distribution with rate parameter $\lambda_g e^{\eta s}$, where $\lambda_g > 0$ and $\eta > 0$.  The rates $\lambda_g$ are unrelated to the instability parameter $\lambda_0$ of Section~\ref{sec:mean-field}; the letter $\lambda$ is standard in both roles.  The baseline $s$ enters by compressing the threshold distribution rather than acting as a strict additive shift.  This captures the empirical observation that institutional pressure does not simply subtract from a fixed resistance level but reshapes the entire distribution of susceptibility, pushing more agents toward the activation boundary.

An exponential distribution is a standard one-parameter continuous distribution on $[0,\infty)$.  Under this distribution, most agents have relatively small thresholds, with progressively fewer agents at large ones. The parameter $\lambda_g$ controls the intrinsic vulnerability of group $g$; a larger $\lambda_g$ means more of group $g$'s agents have small net thresholds and are therefore easier to push into harmful behavior.  The parameter $\eta$ controls how sensitive both groups are to the shared baseline pressure $s$; when $\eta$ is large, even a modest increase in $s$ concentrates many more agents near the activation boundary.

The density of this distribution is $f_g(x;\,s) = \lambda_g\, e^{\eta s}\, \exp(-\lambda_g\, e^{\eta s}\, x)$, which is continuous on $(0,\infty)$.  At $x = 0$, $F_g(0;\,s) = 1 - \exp(0) = 0$, confirming the first regularity condition.  The boundary density is
\begin{equation}
\label{eq:boundary-density}
f_g(0^+;\, s) \;=\; \lambda_g\, e^{\eta s}.
\end{equation}
If $\lambda_B > \lambda_A$---that is, if group $B$ is intrinsically more vulnerable---then $f_B(0^+;\,s) > f_A(0^+;\,s)$ for every baseline $s$.  In words, group $B$ always has more near-threshold agents than group $A$, no matter how much or how little baseline pressure the institution applies.

To illustrate, suppose $\lambda_A = 1$, $\lambda_B = 3$, and $\eta = 0.5$.  At baseline $s = 0$, the boundary densities are $f_A(0^+;\,0) = 1$ and $f_B(0^+;\,0) = 3$, so group $B$ has three times as many near-threshold agents as group $A$.  If the baseline increases to $s = 1$, these become $f_A(0^+;\,1) = e^{0.5} \approx 1.65$ and $f_B(0^+;\,1) = 3e^{0.5} \approx 4.95$.  Both groups become more vulnerable, but the three-to-one ratio persists---the gap between the groups is structural, not a consequence of any particular pressure level.

So far we have described how each group's thresholds are distributed, but we have not said anything about how the two groups influence each other.

\begin{definition}[Two-group interaction structure]
\label{def:two-group-interaction}
The \emph{interaction matrix} is a $2 \times 2$ matrix $M = (m_{gg'})$ with $m_{gg'} \geq 0$, where $m_{gg'}$ is the strength of social influence that group $g'$ exerts on group $g$.  For each group $g \in \{A, B\}$, the \emph{harm rate} $h_g \in [0, 1]$ is the proportion of group-$g$ agents engaging in harmful behavior.  In the vanishing-seed limit, the harm rates satisfy the fixed-point equation
\begin{equation}
\label{eq:two-group-fp}
h_g \;=\; F_g\!\Big(\sum_{g' \in \{A,B\}} m_{gg'}\, h_{g'};\; s\Big), \qquad g \in \{A, B\}.
\end{equation}
The \emph{harm-free state} is the fixed point $h_A = h_B = 0$.
\end{definition}

For example, $m_{AB}$ is the influence that group $B$'s behavior has on group $A$'s members.  The right-hand side of~\eqref{eq:two-group-fp} is the proportion of group-$g$ agents whose threshold is below the total social pressure $\sum_{g'} m_{gg'} h_{g'}$; at a fixed point, this proportion must equal the harm rate $h_g$ itself.  The harm-free state is always a fixed point, since $F_g(0;\,s) = 0$ and no one activates when there is no social pressure.  The central question is whether this fixed point is \emph{stable}.  If a tiny fraction of agents in each group begins engaging in harmful behavior, does the perturbation die out or grow?

The stability question becomes a problem about a two-dimensional discrete map.  At each time step, the pair $(h_A, h_B)$ updates according to the fixed-point equations above, and we need to determine whether small perturbations of $(0,0)$ grow or shrink.  Standard dynamical-systems theory linearizes around the harm-free state: when harm rates $h_A$ and $h_B$ are small, the nonlinear fixed-point equations are well approximated by a linear system whose behavior is governed by a single matrix, the \emph{Jacobian}.

Concretely, each fixed-point equation has the form $h_g = F_g(\text{something depending on } h_A, h_B)$.  To linearize, we differentiate with respect to $h_{g'}$ and evaluate at $h_A = h_B = 0$.  By the chain rule, differentiating $F_g\big(\sum_{\ell \in \{A,B\}} m_{g\ell} h_{\ell};\,s\big)$ with respect to $h_{g'}$ pulls down the density $f_g$ times the coefficient $m_{gg'}$.  Evaluating at the origin gives
\begin{equation}
\frac{\partial}{\partial h_{g'}}\, F_g\!\Big(\sum_{\ell \in \{A,B\}} m_{g\ell}\, h_{\ell};\; s\Big)\bigg|_{h_A = h_B = 0} \;=\; f_g(0^+;\, s) \;\cdot\; m_{gg'}.
\end{equation}
Collecting these partial derivatives into a matrix, we obtain the Jacobian
\begin{equation}
\label{eq:jacobian}
J(s) \;=\; \mathrm{diag}\!\big(f_A(0^+;\, s),\; f_B(0^+;\, s)\big) \;\cdot\; M.
\end{equation}
The Jacobian $J(s)$ encodes both the near-threshold density of each group (through the diagonal factor) and the inter-group influence structure (through $M$).

Standard linearized-stability theory for discrete maps gives the stability criterion in terms of the eigenvalues of $J(s)$: the harm-free fixed point is stable when every eigenvalue has magnitude below 1 and unstable when any eigenvalue exceeds 1.  Since $f_g(0^+;\,s) \geq 0$ and $m_{gg'} \geq 0$, the Jacobian has nonnegative entries, and the Perron--Frobenius theorem identifies a single dominant eigenvalue, real and nonnegative and of maximal magnitude, that controls stability on its own.  This dominant eigenvalue plays the role of a \emph{cascade threshold}, the two-group analogue of $\lambda_0 = \alpha f(0^+)$ from Section~\ref{sec:mean-field}.

As a first application of the Jacobian framework, consider the simplest possible interaction structure, in which both groups experience the same aggregate social pressure $\alpha H$, where $\alpha > 0$ is a coupling strength and $H$ is the overall population-wide harm rate.  This is the \emph{scalar-contagion} case, where there is only one channel of social influence, and it affects everyone equally.  In this setting, the cascade threshold for each group is $\alpha\, f_g(0^+;\,s)$, and we can define a \emph{robustness margin} that measures how far each group is from instability.

\begin{proposition}[Impossibility of neutral equal robustness]
\label{thm:impossibility}
In the scalar-contagion model with the parametric family~\eqref{eq:parametric}, if $\lambda_B > \lambda_A$, then for every group-neutral baseline $s$ the robustness margins
\begin{equation}
\label{eq:robustness-margin}
R_g(s) \;=\; 1 - \alpha\, f_g(0^+;\, s) \;=\; 1 - \alpha\, \lambda_g\, e^{\eta s}
\end{equation}
satisfy $R_B(s) < R_A(s)$.
\end{proposition}

A robustness margin above zero means the group would be stable in isolation (perturbations die out); a margin below zero means the group would be unstable in isolation (cascading if decoupled from the other group).  The proposition says that no uniform baseline $s$ can equalize the two groups' robustness margins.  Group $B$ is always closer to the brink.  The conclusion does not depend on the exponential family specifically; it holds for any threshold distributions satisfying $f_B(0^+;\,s) > f_A(0^+;\,s)$ for all admissible $s$, of which the exponential family is one concrete example.

\begin{proof}
Since $f_g(0^+;\,s) = \lambda_g e^{\eta s}$, we have $R_g(s) = 1 - \alpha \lambda_g e^{\eta s}$.  Because $\lambda_B > \lambda_A$, it follows that $\alpha \lambda_B e^{\eta s} > \alpha \lambda_A e^{\eta s}$, and therefore $R_B(s) < R_A(s)$ for every $s$.
\end{proof}

In practice, agents in an institution are often more influenced by members of their own group than by members of the other group.  A junior employee is more likely to take behavioral cues from other junior employees than from senior administrators, and vice versa.  This tendency is called \emph{homophily}---literally, ``love of the same''---and is one of the best-documented patterns in social-network research \cite{mcpherson2001}.

\begin{definition}[Homophily parametrization]
\label{def:homophily}
For each group $g \in \{A, B\}$, let $\beta_g > 0$ denote the \emph{total influence capacity} of group $g$, the overall strength of social pressure that group $g$ can exert on others.  The \emph{homophily parameter} $\pi \in [1/2, 1]$ is the proportion of social influence directed within one's own group.  The interaction matrix under homophily $\pi$ is
\begin{equation}
\label{eq:homophily-matrix}
M(\pi) \;=\; \begin{pmatrix} \pi\,\beta_A & (1-\pi)\,\beta_B \\ (1-\pi)\,\beta_A & \pi\,\beta_B \end{pmatrix}.
\end{equation}
\end{definition}

When $\pi = 1/2$, influence is equally split between groups (no homophily); as $\pi \to 1$, influence becomes entirely within-group (complete homophily).  Reading across the first row of $M(\pi)$, group $A$ receives within-group influence of strength $\pi\beta_A$ (a proportion $\pi$ of group $A$'s total influence capacity directed inward) and between-group influence of strength $(1-\pi)\beta_B$ (a proportion $1-\pi$ of group $B$'s capacity directed outward).  The second row says the same for group $B$, with the roles reversed.

Combining the Jacobian formula~\eqref{eq:jacobian} with the interaction matrix $M(\pi)$, the stability of the harm-free state depends on the eigenvalues of
\begin{equation}
J(s,\pi) \;=\; \mathrm{diag}\!\big(f_A(0^+;\, s),\; f_B(0^+;\, s)\big) \;\cdot\; M(\pi).
\end{equation}
Before computing these eigenvalues, it helps to introduce a shorthand that has a natural interpretation.  Define group $g$'s \emph{within-group instability parameter}:
\begin{equation}
x_g(s) \;=\; f_g(0^+;\, s) \;\cdot\; \beta_g \;=\; \lambda_g\, \beta_g\, e^{\eta s}.
\end{equation}
This product combines two things, namely how many of group $g$'s agents are near the activation threshold ($f_g(0^+;\,s)$) and how much total influence group $g$ exerts ($\beta_g$).  The interpretation is direct.  $x_g(s)$ measures how close group $g$ is to cascading on its own.  When $x_g < 1$, group $g$ would be stable in isolation; when $x_g > 1$, group $g$ would cascade even without any influence from the other group.

Using this shorthand, the Jacobian becomes
\begin{equation}
J(s,\pi) \;=\; \begin{pmatrix} \pi\, x_A & (1-\pi)\, f_A(0^+;\,s)\, \beta_B \\ (1-\pi)\, f_B(0^+;\,s)\, \beta_A & \pi\, x_B \end{pmatrix}.
\end{equation}
The diagonal entries $\pi\, x_g$ come from within-group influence: $f_g(0^+;\,s) \cdot \pi\beta_g = \pi\, x_g$.  The off-diagonal entries $f_g(0^+;\,s) \cdot (1-\pi)\beta_{g'}$ come from between-group influence.

The eigenvalues of $J(s, \pi)$ follow from direct computation.  The trace is $\pi(x_A + x_B)$, and the product of the off-diagonal entries is
\begin{equation}
J_{12}\, J_{21} \;=\; (1-\pi)^2\, f_A(0^+;\,s)\, \beta_B \;\cdot\; f_B(0^+;\,s)\, \beta_A \;=\; (1-\pi)^2\, x_A\, x_B,
\end{equation}
where the last step regroups the factors using the definition $x_g = f_g(0^+;\,s)\,\beta_g$.  Substituting into the standard $2 \times 2$ eigenvalue formula gives
\begin{equation}
\label{eq:eigenvalue}
\mu(s, \pi) \;=\; \frac{\pi(x_A + x_B) \;+\; \sqrt{\pi^2(x_A - x_B)^2 + 4(1-\pi)^2\, x_A\, x_B}}{2}.
\end{equation}
The harm-free state is stable when $\mu < 1$ and unstable when $\mu > 1$.

To see how homophily affects stability, suppose $x_A = 0.6$ and $x_B = 1.3$.  Group $B$ would cascade on its own ($x_B > 1$) but group $A$ would not ($x_A < 1$).  The simple average is $\frac{1}{2}(x_A + x_B) = 0.95$, which is below 1---so if the two groups were perfectly mixed, the institution would be stable overall.

At low homophily ($\pi = 1/2$), the eigenvalue formula gives $\mu = 0.95$, confirming that the institution is indeed stable, because cross-group mixing dilutes group $B$'s instability.  But at high homophily ($\pi = 0.8$), the eigenvalue rises to $\mu \approx 1.09$ because group $B$ is now effectively isolated from the stabilizing influence of group $A$, and its instability dominates.  There is a critical homophily level $\bar{\pi} \approx 0.65$; above this threshold, the harm-free state is unstable, and a cascade of harmful behavior is possible.

The next result formalizes this pattern.  It shows that whenever one group is individually unstable but the institution is stable on average, increasing homophily eventually destabilizes the system, and it gives an explicit formula for the tipping point.  Throughout, the structural parameters and the baseline pressure $s$ are fixed; the variable of interest is $\pi$.  For each fixed $s$, the products $x_g(s) = f_g(0^+; s) \cdot \beta_g$ are constants, and $\mu(s, \pi)$ reduces to a single-variable function of $\pi$ on $[1/2, 1]$.

\begin{theorem}[Homophily destabilizes neutral safety]
\label{thm:homophily}
In the parametric two-group model, suppose
\begin{equation}
\label{eq:homophily-hyp}
x_B(s) > 1 > x_A(s), \qquad x_A(s) + x_B(s) < 2.
\end{equation}
The first condition says that group $B$ would be individually unstable ($x_B > 1$) while group $A$ would be individually stable ($x_A < 1$); the second condition says that on average the institution is still in the stable regime.  Together, these describe a setting in which safety depends on inter-group mixing dampening the instability of the more exposed group.  Then:
\begin{enumerate}[label=(\roman*)]
    \item At low homophily ($\pi = 1/2$), the harm-free state is stable: $\mu(s, 1/2) < 1$.
    \item At full homophily ($\pi = 1$), it is unstable: $\mu(s, 1) > 1$.
    \item There exists at least one crossing value $\bar{\pi}(s) \in (1/2, 1)$ at which $\mu(s, \bar{\pi}) = 1$.  Solving $\mu = 1$ yields the explicit formula
\begin{equation}
\label{eq:critical-pi}
\bar{\pi}(s) = \frac{1 - x_A(s)\, x_B(s)}{x_A(s) + x_B(s) - 2\, x_A(s)\, x_B(s)},
\end{equation}
which satisfies $1/2 < \bar{\pi}(s) < 1$ under the stated hypotheses.
    \item The eigenvalue $\mu(s, \pi)$ is strictly increasing in $\pi$ on $[1/2,\, 1]$.  Consequently, the crossing $\bar{\pi}(s)$ is unique, and $\mu(s,\pi) > 1$ for all $\pi > \bar{\pi}(s)$.  Once homophily exceeds the critical level, the harm-free state is unstable.
\end{enumerate}
\end{theorem}

Part~(iv) gives the theorem its sharpest interpretation.  Homophily has a monotone destabilizing effect, and the critical level $\bar{\pi}(s)$ is a clean transition point rather than merely one of possibly several crossings.  The monotonicity also has a practical consequence.  There is no ``safe'' level of moderate homophily that improves institutional stability.  A policymaker cannot hope that within-group affinity helps groups police their own behavior; in this model, every increment of homophily makes the institution less safe, never more.

\begin{proof}
\emph{Part (i): stability at low homophily.}  At $\pi = 1/2$, the term under the square root in~\eqref{eq:eigenvalue} simplifies: $\tfrac{1}{4}(x_A - x_B)^2 + x_A x_B = \tfrac{1}{4}(x_A + x_B)^2$.  So the square root equals $\tfrac{1}{2}(x_A + x_B)$, and the numerator becomes $\tfrac{1}{2}(x_A + x_B) + \tfrac{1}{2}(x_A + x_B) = x_A + x_B$, giving $\mu = (x_A + x_B)/2 < 1$ by hypothesis.

\emph{Part (ii): instability at full homophily.}  At $\pi = 1$, the Jacobian $J$ becomes diagonal with entries $x_A(s)$ and $x_B(s)$; since $x_B(s) > 1$, the largest eigenvalue exceeds 1.

\emph{Part (iii): existence of a crossing.}  Since $\mu(s, \cdot)$ is continuous on $[1/2, 1]$ with $\mu(s, 1/2) < 1$ and $\mu(s, 1) > 1$, the intermediate value theorem gives at least one $\bar{\pi} \in (1/2, 1)$ with $\mu(s, \bar{\pi}) = 1$.

\emph{Derivation of the formula.}  The key simplification is that setting $\mu = 1$ in the eigenvalue formula and squaring both sides
yields, after expansion and cancellation, a linear equation in $\pi$ that can be solved in closed form.  Set $\mu(s, \pi) = 1$ in~\eqref{eq:eigenvalue}.  The eigenvalue equation becomes
\begin{equation}
\pi(x_A + x_B) - 2 \;=\; -\sqrt{\pi^2(x_A - x_B)^2 + 4(1-\pi)^2 x_A x_B}.
\end{equation}
Since $x_A + x_B < 2$ and $\pi \leq 1$, the left-hand side is negative, matching the sign of the right-hand side, so squaring does not introduce a spurious solution from the opposite-sign branch.
Squaring both sides:
\begin{align}
\text{LHS}^2 &= \pi^2(x_A + x_B)^2 - 4\pi(x_A + x_B) + 4, \\
\text{RHS}^2 &= \pi^2(x_A - x_B)^2 + 4(1-\pi)^2 x_A x_B.
\end{align}
Expanding $(x_A + x_B)^2 - (x_A - x_B)^2 = 4x_A x_B$ and $(1-\pi)^2 = 1 - 2\pi + \pi^2$, then canceling:
\begin{equation}
4\pi^2 x_A x_B - 4\pi(x_A + x_B) + 4 \;=\; 4(1 - 2\pi + \pi^2) x_A x_B.
\end{equation}
Simplifying gives $\pi(x_A + x_B) = 1 + 2\pi x_A x_B - x_A x_B$, i.e.,
\begin{equation}
\pi\big[x_A + x_B - 2 x_A x_B\big] \;=\; 1 - x_A x_B.
\end{equation}
Write $a = x_A(s)$, $b = x_B(s)$, and $D = a + b - 2ab$.  The hypotheses give $a \in (0, 1)$, $b > 1$, and $a + b < 2$, defining a triangular region in the $(a, b)$-plane.  By the arithmetic mean--geometric mean inequality, $ab \leq (a+b)^2/4 < (a+b)/2$, so $D > 0$ throughout this region.  Dividing by $D$ yields~\eqref{eq:critical-pi}.
\emph{The formula lies between $1/2$ and $1$.}  Differentiating $\bar{\pi} = (1-ab)/D$ gives
\begin{equation}
\frac{\partial \bar{\pi}}{\partial a} \;=\; -\frac{(b-1)^2}{D^2}, \qquad \frac{\partial \bar{\pi}}{\partial b} \;=\; -\frac{(a-1)^2}{D^2},
\end{equation}
both strictly negative on the interior.  So $\bar{\pi}$ is strictly decreasing in both $a$ and $b$, and the extreme values occur on the boundary.  On the edge $a = 0$ with $b \in (1, 2)$, $\bar{\pi}(0, b) = 1/b \in (1/2, 1)$.  On the edge $b = 1$ with $a \in (0, 1)$, the $(1-a)$ factor cancels and $\bar{\pi}(a, 1) = 1$. On the hypotenuse $a + b = 2$ with $a \in (0, 1)$, substituting $b = 2-a$ gives $\bar{\pi} = (1-a)^2/[2(1-a)^2] = 1/2$.
Since the open hypothesis region lies strictly inside these boundary edges, $1/2 < \bar{\pi}(s) < 1$.

\emph{Part (iv): strict monotonicity.}  Write $a = x_A(s)$ and $b = x_B(s)$, and define
\begin{equation}
\Delta(\pi) \;=\; \sqrt{\pi^2(a-b)^2 + 4(1-\pi)^2\, ab}.
\end{equation}
Differentiating $\mu = \bigl(\pi(a+b) + \Delta\bigr)/2$ gives
\begin{equation}
2\Delta \cdot \mu'(\pi) \;=\; (a+b)\,\Delta \;+\; \pi(a-b)^2 \;-\; 4(1-\pi)\,ab.
\end{equation}
Since $\pi^2(a-b)^2 \geq 0$, we have $\Delta \geq 2(1-\pi)\sqrt{ab}$.  Substituting this bound, the right-hand side is bounded below by
\begin{equation}
\pi(a-b)^2 \;+\; 2(1-\pi)\sqrt{ab}\,\bigl[(a+b) - 2\sqrt{ab}\,\bigr] \;=\; \pi(a-b)^2 \;+\; 2(1-\pi)\sqrt{ab}\,(\sqrt{a}-\sqrt{b})^2.
\end{equation}
Both terms are nonnegative, and the first is strictly positive because $a \neq b$ (since $b > 1 > a$).  Therefore $\mu'(\pi) > 0$ on $[1/2,\, 1]$.  Combined with parts~(i)--(iii), the unique crossing at $\bar{\pi}$ separates a stable regime ($\pi < \bar{\pi}$) from an unstable one ($\pi > \bar{\pi}$).
\end{proof}

\section{Institutional Repair}
\label{sec:repair}

The preceding sections diagnosed when harm spreads and how signed influence and group structure interact with cascade risk.  Suppose the diagnosis is bad---some agents are so susceptible to peer pressure, given the institution's current structure, that the robustness condition fails.  What can the institution do about it?

An institution has two levers.  The first is to change the \emph{environment}: reduce the baseline pressures that push agents toward harm in the first place.  In practice, this might mean revising incentive structures, reducing workload, or removing norms that normalize harmful conduct.  Formally, we model this as \emph{baseline reduction}: decrease the baseline pressure $\sigma_i$ on agent $i$ by an amount $u_i \geq 0$.  The modified baseline is $\sigma_i - u_i$, and the modified net threshold is
\begin{equation}
\tau_i - (\sigma_i - u_i) = (\tau_i - \sigma_i) + u_i = \theta_i + u_i.
\end{equation}
So reducing baseline by $u_i$ increases the net threshold by $u_i$, making the agent harder to push into harm.

The second lever is to restructure \emph{relationships}: reduce the influence that one agent's harmful behavior exerts on another.  In practice, this might mean reassigning reporting lines, limiting exposure between certain roles, or introducing structural buffers.  Formally, we model this as \emph{influence reduction}: decrease the influence weight $w_{ij}$ by an amount $z_{ij} \geq 0$, giving modified influence $w_{ij} - z_{ij}$.  This reduces the social pressure that agent $j$'s harmful behavior places on agent $i$.

Institutional change is costly.  Sunk costs in existing structures, disruption of political coalitions, and threats to institutional legitimacy have been identified as sources of \emph{structural inertia} that resist reform \cite{hannan1984}.  On the algorithmic side, influence maximization in cascade models is NP-hard but admits greedy approximation guarantees under standard diffusion models, framing the problem as a combinatorial selection with a cardinality budget \cite{kempe2003}.  Our repair problem differs in that the decision variables are continuous, specifying how much to reduce each baseline or influence weight.  We model intervention costs with per-unit or, more generally, convex cost functions, a standard choice in mathematical programming.  The goal is to minimize the total cost of repair subject to the constraint that the repaired institution is $k$-robust.

The $k$-robustness condition from Theorem~\ref{thm:topk} is a \emph{strict} inequality: $T_k(i) < \theta_i$.  Strict inequalities create a real problem for optimization.  To see why, consider the simplest possible example: a single agent sits exactly at the boundary, $T_k(i) = \theta_i$, and we repair by increasing the net threshold by $u$ at linear cost $c(u) = u$.  Robustness requires $T_k(i) < \theta_i + u$, which means $u > 0$.  But there is no \emph{smallest} positive number.  We could choose $u = 0.1$, or $u = 0.01$, or $u = 0.001$, each one cheaper than the last, but no feasible repair achieves the infimum cost of zero.  The minimum-cost repair does not exist.

The same issue arises in engineering.  A bridge designed to hold \emph{exactly} the expected load is fragile, and any modeling error could push it past the breaking point.  Engineers solve this by requiring a safety margin, and we do the same.  We introduce an explicit \emph{design margin} $\delta > 0$, chosen by the institutional designer, and require
\begin{equation}
\label{eq:delta-robust}
T_k(i) \;\leq\; \theta_i - \delta
\end{equation}
for every agent $i$.  We call an institution satisfying~\eqref{eq:delta-robust} \emph{$(k,\delta)$-robust}, meaning it can survive any $k$-seed shock with a safety buffer of $\delta$ to spare.  The margin $\delta$ serves two purposes.  First, it represents a genuine safety buffer that absorbs small perturbations in parameters or personnel.  Second, it closes the feasible set, replacing the strict inequality $<$ with the non-strict $\leq$, which guarantees that a minimum-cost repair exists.

\subsection{Baseline-only repair}

We begin with the simplest case.  Suppose the institution adjusts only baseline pressures, leaving all influence relationships unchanged.  Recall that reducing the baseline pressure on agent $i$ by $u_i$ increases the net threshold from $\theta_i$ to $\theta_i + u_i$.  The $(k,\delta)$-robustness condition for agent $i$ after repair is therefore
\begin{equation}
T_k(i) \;\leq\; (\theta_i + u_i) - \delta, \qquad\text{i.e.,}\qquad u_i \;\geq\; T_k(i) - \theta_i + \delta.
\end{equation}
In words, the required repair for agent $i$ is the amount by which the worst-case pressure exceeds the current safety margin, plus the design buffer $\delta$.  If $T_k(i) - \theta_i + \delta \leq 0$, then agent $i$ is already safe and no repair is needed.

The cost of reducing baselines will generally differ from one agent to another: changing incentives for a senior manager may require different resources than adjusting workload for a frontline employee.  We model this with a per-unit cost $c_i > 0$ for each agent $i$, so that reducing agent $i$'s baseline by $u_i$ costs $c_i \, u_i$.  This linear cost model is the simplest reasonable assumption; the mixed-repair formulation in Section~\ref{sec:mixed} relaxes it to allow convex costs.

\begin{proposition}[Optimal baseline-only repair]
\label{thm:baseline-repair}
Fix a design margin $\delta > 0$.  With per-agent linear costs $c_i > 0$ and the constraint $0 \leq u_i \leq \sigma_i$, the minimum-cost baseline-only repair for $(k,\delta)$-robustness is
\begin{equation}
u_i^* = \max\!\big\{0,\; T_k(i) - \theta_i + \delta\big\}, \qquad \text{total cost} \;\;= \sum_i c_i\, u_i^*,
\end{equation}
provided $u_i^* \leq \sigma_i$ for every agent $i$.  If $T_k(i) - \theta_i + \delta > \sigma_i$ for some agent, baseline-only repair is infeasible for that agent.  The problem is separable across agents: each agent's optimal repair depends only on their own $T_k(i)$, $\theta_i$, and $\sigma_i$.
\end{proposition}

\begin{proof}
Since $T_k(i)$ depends only on $W$, which is unchanged under baseline-only repair, and $u_i$ affects only agent $i$'s net threshold, there is no coupling between agents: each agent's repair problem is independent.  For each agent $i$, the constraint $u_i \geq T_k(i) - \theta_i + \delta$ with linear cost $c_i u_i$ is minimized by taking $u_i$ as small as possible, which gives $u_i^* = \max\{0,\; T_k(i) - \theta_i + \delta\}$.
\end{proof}

Because the problem is separable, the institution can assess and repair each agent independently, without worrying about interactions between repairs.  This makes baseline-only repair computationally straightforward even for large institutions.

\subsection{Influence-only repair}

Now suppose the institution restructures influence relationships, leaving baselines unchanged.  For each pair of agents, the institution can reduce the influence weight $w_{ij}$ by an amount $z_{ij} \geq 0$, at some cost.  Since the $(k,\delta)$-robustness constraint for agent $i$ depends only on the influence weights coming \emph{into} agent $i$---that is, row $i$ of $W$---the constraints decompose agent by agent.

Whether the \emph{optimization} also decomposes depends on how costs are structured.  If the cost of reducing each weight $w_{ij}$ depends only on that weight and no others, then the full problem splits into $n$ independent subproblems, one per agent.  A shared global budget $\sum_{i,j} z_{ij} \leq B$ would link the subproblems, but we do not impose one here.  Such a budget would preserve convexity and could be handled by Lagrangian decomposition.  Even with separable costs, using both levers simultaneously is valuable: for some agents, reducing baseline pressure is cheaper, while for others, restructuring influence is cheaper.  The mixed formulation captures this agent-by-agent flexibility.

\subsection{Mixed repair}
\label{sec:mixed}

In practice, an institution would use both levers at once, reducing some baselines and restructuring some influence relationships simultaneously.  The formulation below optimizes for safety alone; it does not impose a coordination-preservation constraint analogous to $\mathcal{G}(r,\Lambda)$ from Section~\ref{sec:functionality}.  The optimizer may therefore achieve safety at the expense of functionality.  Incorporating a coordination floor is a natural extension but is not needed for the convexity result.  The question is whether the safety-only problem remains tractable.

Write $T_k^{(i)}(W - Z)$ for the top-$k$ incoming influence sum of agent $i$ after the influence matrix has been reduced from $W$ to $W - Z$.  The cost of reducing agent $i$'s baseline by $u_i$ is $c_i(u_i)$, and the cost of reducing the influence of $j$ on $i$ by $z_{ij}$ is $d_{ij}(z_{ij})$.  We allow these cost functions to be any convex, nondecreasing functions, finite on $[0, \infty)$ and hence continuous---linear costs are a special case, but convex costs also capture settings where larger interventions are disproportionately expensive.  The $(k,\delta)$-robust repair problem is:
\begin{equation}
\label{eq:mixed-repair}
\min_{u, Z} \;\; \sum_{i} c_i(u_i) + \sum_{i,j} d_{ij}(z_{ij})
\end{equation}
subject to
\begin{equation}
T_k^{(i)}(W - Z) \;\leq\; \theta_i + u_i - \delta \quad \text{for all } i, \qquad 0 \leq z_{ij} \leq w_{ij}, \qquad 0 \leq u_i \leq \sigma_i.
\end{equation}

The constraints say that, after repair, every agent must be $(k,\delta)$-robust.  The bounds $0 \leq z_{ij} \leq w_{ij}$ mean we can reduce influence weights but not make them negative or add new ones.  The bound $u_i \leq \sigma_i$ means we cannot reduce baseline pressure below zero.  With this bound, the problem may be infeasible: if the required robustness gap exceeds what can be closed by eliminating all baseline pressure and all influence weights, no repair exists within the model.  When the problem is feasible, the following theorem guarantees an optimal solution.

A small example illustrates the tradeoff.  Suppose there are two agents with $k = 1$ and $\delta = 0.1$.  Agent 1 has $T_1(1) = 0.9$ and $\theta_1 = 0.7$, so the robustness gap is $T_1(1) - \theta_1 + \delta = 0.3$.  The institution can close this gap by raising agent 1's net threshold ($u_1 = 0.3$), by reducing the largest incoming influence weight ($z_{1j} = 0.3$), or by some combination.  If baseline reduction costs $\$2$ per unit and influence reduction costs $\$5$ per unit, pure baseline repair costs $\$0.60$ while pure influence repair costs $\$1.50$; the optimizer would favor baseline reduction for this agent.  For a different agent where influence reduction is cheaper, the optimizer would choose the other lever.  The mixed formulation finds the cheapest combination of both levers for each agent.  Because the objective and constraints decompose by agent and no shared budget links the subproblems, the mixed repair problem is perfectly separable into $n$ independent subproblems, just as in the baseline-only case.

The reason to expect convexity is structural.  The robustness constraints are built from top-$k$ sums.  A top-$k$ sum is the maximum of finitely many linear functions, one for each way to choose $k$ entries.  A maximum of linear functions is convex, and a constraint of the form ``convex function $\leq$ linear function'' carves out a convex region.

\begin{theorem}[Convexity of mixed repair]
\label{thm:mixed-repair}
If $c_i$ and $d_{ij}$ are convex and nondecreasing, then the $(k,\delta)$-robust mixed repair problem~\eqref{eq:mixed-repair} is a convex optimization problem.  Any local minimum is also a global minimum, and off-the-shelf convex solvers can compute optimal repairs.  Whenever the feasible set is nonempty, the minimum cost is achieved by some repair, not just approached as a limit.
\end{theorem}

\begin{proof}
The argument has four steps: show the top-$k$ sum is convex, show the constraints define a convex feasible set, show the objective is convex, and verify the minimum is attained.

\emph{Step 1: the top-$k$ sum is convex.}  The top-$k$ sum of a vector $v = (v_1, \ldots, v_m)$ is the sum of its $k$ largest entries.  We can write this as
\begin{equation}
T_k(v) = \max\left\{\sum_{j=1}^m s_j v_j \;:\; s \in \{0,1\}^m,\; \sum_j s_j = k\right\},
\end{equation}
where the maximum ranges over all ways to select $k$ entries.  For any fixed selection $s$, the expression $\sum_j s_j v_j$ is linear in $v$.  The top-$k$ sum is the maximum over finitely many linear functions, and a maximum of linear functions is always convex.

\emph{Step 2: the constraints define a convex feasible set.}  The constraint for agent $i$ requires $T_k^{(i)}(W - Z) \leq \theta_i + u_i - \delta$.  The left side, $T_k^{(i)}(W - Z)$, is a convex function of the reduction matrix $Z$ because subtracting $Z$ from $W$ is a linear operation and composing a convex function with a linear operation preserves convexity.  The right side, $\theta_i + u_i - \delta$, is linear in $u_i$.  So each constraint asks that a convex function be bounded by a linear function, which carves out a convex region.  The box constraints $0 \leq z_{ij} \leq w_{ij}$ and $u_i \geq 0$ are also convex.  The intersection of all these constraints is therefore convex.  Because every constraint uses $\leq$ rather than $<$, the feasible set is also closed.

\emph{Step 3: the objective is convex.}  Each $c_i(u_i)$ and $d_{ij}(z_{ij})$ is convex by assumption, and a sum of convex functions is convex.

A convex objective over a convex feasible set has no spurious local minima, so every local minimum is also a global minimum.

\emph{Step 4: attainment.}  We show that, when the feasible set is nonempty, the minimum cost is achieved.  The feasible region for $(u, Z)$ is defined by box constraints $0 \leq u_i \leq \sigma_i$ and $0 \leq z_{ij} \leq w_{ij}$ together with closed convex constraints.  This region is a closed subset of a bounded box, hence compact whenever it is nonempty.  In finite dimensions, a continuous function attains its minimum on a compact set, so the objective attains its minimum here.
\end{proof}

This convexity is a structural advantage over generic network intervention problems, which typically require integer programming.  It arises because absorbing dynamics make first-step prevention sufficient.  Theorem~\ref{thm:topk} tells us that the institution is safe if and only if no agent is pushed past their threshold by the worst-case coalition, and this condition involves only top-$k$ sums of the influence matrix rather than multi-step cascade simulations.  The design margin $\delta$ closes the feasible set, and the compactness of the influence-reduction box $0 \leq Z \leq W$ ensures that the optimal repair exists.

The same convexity extends to the signed influence model of Section~\ref{sec:signed}.  Recall from Section~\ref{sec:signed-convexity} that the signed top-$k$ sum $\Phi_k$ is convex.  The signed repair problem replaces $T_k$ with $\Phi_k$ and allows the institution to reduce entries of the signed influence matrix $\mathcal{A}$.  In this worst-case formulation, entries with $a_{ij} \leq 0$ already contribute zero to $\Phi_k$, and making them more negative does not improve the robustness constraint.  With nondecreasing intervention costs, any feasible repair can therefore be modified, without increasing its cost, so that $y_{ij} = 0$ whenever $a_{ij} \leq 0$ and $y_{ij} \leq a_{ij}$ whenever $a_{ij} > 0$; in particular, when an optimal repair exists, one of this form exists.  The signed repair problem should be read as weakening net harm-promoting influence, not as spending repair effort to strengthen already protective channels.

\begin{corollary}[Signed mixed repair is convex]
\label{cor:signed-repair}
If the cost functions $c_i$ and $d_{ij}$ are convex and nondecreasing, then the signed mixed repair problem
\begin{equation}
\min_{u, Y} \;\; \sum_i c_i(u_i) + \sum_{i,j} d_{ij}(y_{ij})
\end{equation}
subject to
\begin{equation}
\Phi_k(\mathcal{A}_{i\cdot} - Y_{i\cdot}) \;\leq\; \theta_i + u_i - \delta \quad \text{for all } i, \qquad Y \geq 0, \qquad 0 \leq u_i \leq \sigma_i,
\end{equation}
is a convex optimization problem.  Whenever the feasible set is nonempty, the minimum cost is attained provided each influence-intervention cost $d_{ij}$ is coercive, meaning $d_{ij}(y) \to \infty$ as $y \to \infty$.
\end{corollary}

\begin{proof}
The argument follows the same four steps as Theorem~\ref{thm:mixed-repair}.  \emph{Step 1:} $\Phi_k$ is convex, as established in Section~\ref{sec:signed-convexity}.  \emph{Step 2:} each constraint asks that a convex function of $(u, Y)$ be bounded by a linear function, carving out a convex region.  The intersection of all such constraints with the box constraints $Y \geq 0$, $u \geq 0$ is closed and convex.  \emph{Step 3:} the objective is a sum of convex functions, hence convex.  \emph{Step 4:} when each $d_{ij}$ is coercive in $y_{ij}$, every minimizing sequence has bounded $Y$-components; intersecting a bounded sublevel set with the closed feasible region gives a compact set, and the continuous convex objective attains its minimum there.
\end{proof}

\section{Conclusion}
\label{sec:future}

We developed a mathematical framework for studying how harmful behavior can spread through institutions under absorbing threshold dynamics, in which the same influence channels that enable coordination also transmit harm.  The main results are:

\begin{itemize}

\item \textbf{Exact $k$-robustness criterion} (Theorem~\ref{thm:topk}).  Under absorbing dynamics, an institution is safe against any $k$ initially activated agents if and only if, for every agent, the worst-case first-step pressure from the $k$ strongest peers falls below that agent's net resistance threshold.  Because first-step prevention is sufficient under progressive dynamics, this completely characterizes when cascades can and cannot start.

\item \textbf{Signed influence extension} (Theorem~\ref{thm:signed}).  When some peers exert protective pressure that counteracts harmful influence, the same first-step logic applies.  Robustness is determined by whether each agent's worst-case exposure from the $k$ seeds with the largest net positive influence stays below their threshold.  Protective effects are netted against the harm-promoting influence of the same agent, and agents whose net influence is non-positive are ignored because an adversary would never seed them.

\item \textbf{Convexity of worst-case exposure} (Proposition~\ref{prop:phik-convex}).  The function measuring each agent's worst-case exposure is convex in the influence weights.  In the continuous repair formulation, the feasible safety region is therefore convex.  The repair theory (Theorem~\ref{thm:mixed-repair}) exploits this structure to formulate repair as a convex optimization problem.

\item \textbf{Mean-field cascade threshold} (Proposition~\ref{prop:scalar}).  In a mean-field approximation where individual influence weights are small but numerous, a single summary statistic plays the role of $R_0$ in epidemiology.  When this statistic exceeds one, the harm-free state is locally unstable and a positive-harm fixed point exists.

\item \textbf{Exposure floor and sharp frontier} (Lemma~\ref{lem:exposure-floor}, Theorem~\ref{thm:functionality}).  Any institution that is functional enough to coordinate its members must expose every agent to a minimum amount of worst-case pressure.  There exists an institutional design that is both functional and safe if and only if this unavoidable exposure stays below the weakest agent's resistance threshold.

\item \textbf{Neutrality gap} (Corollary~\ref{cor:neutrality-gap}).  When the same coordination standard is applied uniformly to all agents, it can be safe for a more insulated group while being provably unsafe for a more exposed group.  Within the model, the gap is not a failure of implementation but a mathematical consequence of applying a uniform standard to groups with different minimum thresholds.

\item \textbf{Impossibility of neutral equal robustness} (Proposition~\ref{thm:impossibility}).  In the mean-field two-group model, if the two groups differ in how easily their members are pushed into harmful behavior, no neutral institution can provide equal robustness to both groups simultaneously.

\item \textbf{Homophily destabilizes neutral safety} (Theorem~\ref{thm:homophily}).  When an institution is stable on average but one group is individually unstable, increasing homophily monotonically destabilizes the institution.  There is a unique critical homophily level, given in closed form, above which the harm-free state is unstable.

\item \textbf{Convex mixed repair} (Theorem~\ref{thm:mixed-repair}, Corollary~\ref{cor:signed-repair}).  Repairing an unsafe institution by adjusting both baseline conditions and influence weights is a convex optimization problem, which extends to institutions with protective pressures.

\item \textbf{Optimal baseline-only repair} (Proposition~\ref{thm:baseline-repair}).  When only baseline conditions can be changed, the optimal repair has a closed-form solution.  Reduce the baseline pressure on each agent by the exact amount needed to bring their worst-case exposure below their threshold.

\end{itemize}

The sociological literature on institutional harm---from Milgram's obedience experiments to Vaughan's analysis of the Challenger disaster to Browning's study of Reserve Police Battalion 101---has produced detailed accounts of how institutions push people into harmful behavior.  What it has not produced are formal conditions that can be checked \emph{before} a cascade happens.  The standard institutional response is retrospective, investigating after the harm and reforming after the scandal.  Within the absorbing threshold model developed here, the results above provide prospective criteria.  If an institution's influence structure and its members' resistance thresholds are known, one can determine whether a harmful cascade is possible, identify which agents are most vulnerable, and compute the cheapest set of structural changes that would restore safety.  Whether these quantities can be estimated from institutional data is a major open challenge (see below), but the formal framework clarifies what would need to be measured and why.  The fairness results show a further tension within the model.  When a uniform coordination standard is applied to groups with different minimum thresholds, the resulting harm disparity is not a contingent empirical outcome.  It is a mathematical consequence of the model's structure, one that persists across every institutional design meeting that standard.

Several questions remain open, which we organize into four themes.

\emph{Extending the fairness theory.}  The sharp frontier of Theorem~\ref{thm:functionality} characterizes when coordination and robustness can coexist under a uniform standard.  A natural next step is to ask whether additional fairness axioms, or heterogeneous coordination requirements that vary by agent, yield further impossibility results.  Such extensions would move toward an axiomatic theory of institutional design.  On the technical side, the mean-field fairness results (Proposition~\ref{thm:impossibility} and Theorem~\ref{thm:homophily}) describe large-population behavior; it remains open whether analogous impossibility results hold in the finite-agent model under appropriate conditions.

\emph{Richer dynamics.}  The current model assumes that all seeds arrive simultaneously and that harmful behavior, once adopted, is permanent.  In the current deterministic absorbing model with a fixed seed budget, sequential seed placement does not create a genuinely new robustness problem.  The $T_k$ criterion already characterizes the worst case over all seed sets of bounded size, so the analysis reduces to worst-case seed selection.  Relaxing either assumption opens new territory.  Endogenous thresholds, where repeated exposure gradually lowers $\tau_i$ and captures normalization of deviance, would produce hysteresis and path dependence.  Reversible activation, where agents can stop engaging in harm, would require analysis of multiple equilibria rather than monotone convergence.  A continuous-time formulation, allowing group-specific baselines $s_g$,
\begin{equation}
\dot{h}_g \;=\; -h_g + F_g\!\Big(\sum_{g'} m_{gg'}\, h_{g'};\; s_g\Big)
\end{equation}
would give access to bifurcation analysis and Lyapunov stability theory.

\emph{Combinatorial repair.}  The convex repair theory applies when intervention variables are continuous.  In practice, some interventions are all-or-nothing.  Either reassign an employee or do not; either dissolve a reporting relationship or keep it.  Under a global budget, such combinatorial repair may be computationally intractable in general, but we conjecture tractability on restricted network structures such as trees and bounded-treewidth graphs, where dynamic programming could apply.

\emph{Empirical grounding.}  Connecting internal activation to downstream harm on external populations would make the fairness results more concrete and policy-relevant.  The most fundamental empirical challenge is estimating the model's inputs---activation thresholds $\tau_i$, baseline pressures $\sigma_i$, and influence weights $W$---from institutional data.

Of these, empirical estimation is the binding constraint.  Every criterion derived here---for robustness, repair, and homophily---is a function of thresholds, baseline pressures, and influence weights, and estimating those quantities from institutional data would turn conditional theorems into working diagnostics.

\appendix

\section{Notation}
\label{app:notation}

For reference, the table below collects all notation used in the paper.

\begingroup
\small
\renewcommand{\arraystretch}{1.35}
\begin{longtable}{@{}p{3.4cm}p{9.6cm}@{}}
\caption{Notation used throughout the paper.}
\label{tab:notation}\\
\toprule
\multicolumn{2}{@{}l}{\textbf{Exact finite-agent model}} \\
\midrule
\endfirsthead
\multicolumn{2}{@{}l}{\textit{Table~\thetable\ continued from previous page}} \\
\toprule
\midrule
\endhead
\bottomrule
\endlastfoot
\rowcolor{black!8}\multicolumn{2}{@{}c}{\emph{Agents and thresholds (Sections~\ref{sec:agents} and~\ref{sec:separation})}} \\[2pt]
$V = \{1,\dots,n\}$ & Index set labeling $n$ agents. \\
$\tau_i$ & Activation threshold of agent $i$: minimum total pressure needed to activate $i$. \\
$\sigma_i$ & Baseline institutional pressure on agent $i$, independent of peer behavior. \\
$\theta_i = \tau_i - \sigma_i$ & Net threshold of agent $i$. \\
$\underline{\theta}$ & Minimum net threshold in the population: $\underline{\theta} = \min_i \theta_i$. \\[4pt]
\rowcolor{black!8}\multicolumn{2}{@{}c}{\emph{Institutional structure (Section~\ref{sec:structure})}} \\[2pt]
$W$, $w_{ij}$ & Nonnegative influence matrix with zero diagonal ($w_{ii} = 0$); $w_{ij}$ is the pressure increase on $i$ when $j$ is active. \\
$\mathcal{S} = (\sigma, W)$ & Institutional structure. \\[4pt]
\rowcolor{black!8}\multicolumn{2}{@{}c}{\emph{Dynamics (Section~\ref{sec:dynamics})}} \\[2pt]
$S_0$ & Seed set: agents activated exogenously at $t=0$. \\
$x_i(t)$ & State of agent $i$ at time $t$ ($1$ if active, $0$ otherwise); $x(t) = (x_1(t), \ldots, x_n(t))$. \\[4pt]
$w_i(t)$ & Total social pressure on agent $i$ at time $t$: $w_i(t) = \sum_j w_{ij}\, x_j(t)$. \\
$x^*(S_0)$ & Final state vector.  Component $x^*_i = 1$ if agent $i$ is active at termination, $0$ otherwise. \\[4pt]
\rowcolor{black!8}\multicolumn{2}{@{}c}{\emph{Robustness and repair (Sections~\ref{sec:topk} and~\ref{sec:repair})}} \\[2pt]
$T_k(i)$ & Top-$k$ incoming influence sum: sum of the $k$ largest off-diagonal entries in row $i$ of $W$. \\
$\delta$ & Design margin in $(k,\delta)$-robust repair. \\
$u_i$, $z_{ij}$ & Repair variables: baseline reduction for $i$ and influence reduction for entry $(i,j)$. \\
$Y$, $y_{ij}$ & Signed-influence repair variables in Corollary~\ref{cor:signed-repair}. \\
$c_i$, $d_{ij}$ & Baseline and influence repair costs (per-unit prices in Proposition~\ref{thm:baseline-repair}; convex nondecreasing functions in Section~\ref{sec:mixed}). \\[4pt]
\rowcolor{black!8}\multicolumn{2}{@{}c}{\emph{Signed influence (Section~\ref{sec:signed})}} \\[2pt]
$\mathcal{A}$, $a_{ij}$ & Signed influence matrix and its entries. \\
$w^+_{ij}$, $w^-_{ij}$, $a^+_{ij}$ & Harm-promoting and protective components, with $a_{ij} = w^+_{ij} - w^-_{ij}$; $a^+_{ij} = \max\{a_{ij}, 0\}$ is the positive part. \\
$\widetilde{T}_k(i)$ & Signed top-$k$ sum using off-diagonal positive parts of signed influences. \\
$\Phi_k$ & Sum of the $k$ largest positive components of a signed vector; gives $\widetilde{T}_k(i)$ when applied to row $i$ of $\mathcal{A}$. \\[4pt]
\rowcolor{black!8}\multicolumn{2}{@{}c}{\emph{Fairness and impossibility (Section~\ref{sec:impossibility})}} \\[2pt]
$G_A$, $G_B$ & Two groups: more insulated ($G_A$) and more exposed ($G_B$). \\
$\underline{\theta}_{G_A}$, $\underline{\theta}_{G_B}$ & Minimum net thresholds in each group. \\
$\mathcal{G}(r, \Lambda)$ & Coordination condition: every agent can be mobilized by $\leq r$ peers exerting total pressure $\geq \Lambda$. \\
$\Gamma_{k,r}(\Lambda)$ & Exposure floor: minimum possible $T_k$ under $\mathcal{G}(r,\Lambda)$.  Equals $(k/r)\Lambda$ when $k < r$ and $\Lambda$ when $k \geq r$. \\[8pt]
\midrule
\multicolumn{2}{@{}l}{\textbf{Mean-field and two-group model}} \\
\midrule
\rowcolor{black!8}\multicolumn{2}{@{}c}{\emph{One-group mean-field (Section~\ref{sec:mean-field})}} \\[2pt]
$H$, $H^*$, $H^*_+$ & Overall harm rate, its fixed-point value, and a positive fixed point of $g_0$. \\
$p$ & Seed fraction, $p = |S_0|/n$. \\
$F$, $f$ & Net-threshold CDF and density. \\
$\alpha$ & Mean-field coupling strength. \\
$g_p(h)$ & Mean-field update map, $g_p(h) := p + (1-p)\,F(\alpha h)$; $g_0$ is the $p = 0$ case. \\
$\lambda_0 = \alpha f(0^+)$ & Mean-field instability parameter (analogue of $R_0$). \\[4pt]
\rowcolor{black!8}\multicolumn{2}{@{}c}{\emph{Two-group structure (Section~\ref{sec:mean-field-fairness})}} \\[2pt]
$A$, $B$ & Two groups, corresponding to $G_A$ and $G_B$. \\
$h_A$, $h_B$ & Group-specific harm rates. \\
$F_g$, $f_g$ & Net-threshold CDF and density for group $g \in \{A,B\}$. \\
$s$ & Common baseline pressure level under neutral design. \\
$\lambda_g$, $\eta$ & Exponential-family parameters: $\lambda_g$ is group vulnerability, $\eta$ is baseline sensitivity. \\
$M = (m_{gg'})$ & $2 \times 2$ interaction matrix; $m_{gg'}$ is influence from group $g'$ on group $g$. \\[4pt]
\rowcolor{black!8}\multicolumn{2}{@{}c}{\emph{Stability and homophily (Section~\ref{sec:mean-field-fairness})}} \\[2pt]
$J(s)$, $J(s,\pi)$ & Jacobian of the mean-field map at the harm-free state; $J(s,\pi)$ includes homophily. \\
$\beta_g$ & Total influence capacity of group $g$. \\
$\pi$ & Homophily parameter: proportion of influence within one's own group. \\
$x_g(s) = f_g(0^+;s)\,\beta_g$ & Within-group instability parameter for group $g$. \\
$\mu(s, \pi)$ & Largest eigenvalue of $J(s,\pi)$ in the homophily analysis. \\
$\bar{\pi}(s)$ & Critical homophily level at which $\mu = 1$. \\
$R_g(s)$ & Robustness margin for group $g$ in the scalar-contagion special case. \\
\end{longtable}
\endgroup

\end{document}